%% file: DireachDraft.tex
\pdfoutput=1
\documentclass[11pt]{article}
\usepackage[margin = 1 in]{geometry}
\usepackage{mathtools}
\usepackage{amsthm,amsmath}
\usepackage{amssymb}
\usepackage{stmaryrd}
\usepackage{algorithm}
\usepackage[noend]{algpseudocode}
\usepackage{caption}
\usepackage{relsize}
\usepackage{lineno}
\usepackage{graphicx}
\newtheorem{theorem}{Theorem}
\newtheorem{lemma}{Lemma}
\newtheorem*{remark}{Remark}
\newtheorem{definition}{Definition}
\newtheorem{cor}{Corollary}
\newtheorem{claim}{Claim}
\newtheorem{case}{Case}
\newtheorem{assumption}{Assumption}

\usepackage{caption}
\newcommand{\direach}{$S \times V$-direachability~}
\newcommand{\sigmalimit}{\sigma^*}
\usepackage{hyperref}

\title{Faster Multi-Source Directed Reachability via Shortcuts and Matrix Multiplication}
\author{Michael Elkin \\ Ben-Gurion University of the Negev, Israel \\ \texttt{elkinm@bgu.ac.il }  \and Chhaya Trehan \\University of Bristol, UK\\ \texttt{chhaya.trehan@bristol.ac.uk} }
 
\begin{document}
\maketitle
\begin{abstract}
Given an $n$-vertex $m$-edge directed graph $G = (V,E)$ and a designated source vertex $s \in V$, let $\mathcal V_s \subseteq V$ be a subset of vertices 
reachable from $s$ in $G$. Given a subset $S \subseteq V$ of $|S| = n^{\sigma}$, for some $0 < \sigma \le 1$, designated sources,
the $S \times V$-\emph{direachability problem} is to compute the sets $\mathcal V_s$ for every $s \in S$.
Known naive algorithms for this problem either run a BFS/DFS exploration separately from every source, and as a result require $O(m \cdot n^{\sigma})$ time, 
or compute the transitive closure $TC(G)$ of the graph $G$ in $\tilde O(n^{\omega})$ time, where $\omega < 2.371552\ldots$ is the matrix multiplication exponent.
Hence, the current state-of-the-art bound for the problem on graphs with $m = \Theta(n^{\mu})$ edges in $\tilde O(n^{\min \{\mu + \sigma, \omega \}})$,
which may be as large as $\tilde O(n^{\min \{ 2 + \sigma, \omega\}})$.

Our first contribution is an algorithm with running time $\tilde O(n^{1 + \tiny{\frac{2}{3}} \omega(\sigma)})$ for this problem on general graphs, where $\omega(\sigma)$ is the
rectangular matrix multiplication exponent (for multiplying an
$n^{\sigma} \times n$  matrix by an $n \times n$ matrix). Using current state-of-the-art estimates on $\omega(\sigma)$, our exponent is better than
$\min \{2 + \sigma, \omega \}$ for $\tilde \sigma \le \sigma \le 0.53$, where $1/3 < \tilde \sigma < 0.3336$ is a universal constant.
For graphs with $m = \theta(n^{\mu})$ edges, for $1 \le \mu \le 2$, our algorithm has running time $\tilde O(n^{\tiny{\frac{1 + \mu + 2 \cdot \omega(\sigma)}{3}}})$. 
For every $\tilde \sigma \le \sigma < 1$, there exists a non-empty interval
$I_{\sigma} \subseteq [1,2]$, so that our running time is better than the state-of-the-art one whenever the input graph has $m = \Theta(n^{\mu})$ edges with $\mu \in I_{\sigma}$.

Our second contribution is a sequence of algorithms $\mathcal A_0, \mathcal A_1, \mathcal A_2, \ldots$ for the $S \times V$-direachability problem, where $\mathcal A_0$
is our algorithms that we discussed above. We argue that under a certain assumption that we introduce,
 for every $\tilde \sigma \le \sigma < 1$, there exists a sufficiently large index $k = k(\sigma)$ so that $\mathcal A_k$ improves upon the current
state-of-the-art bounds for $S \times V$-direachability with $|S| = n^{\sigma}$, in the densest regime $\mu =2$. (Unconditionally, $\mathcal A_0$ satisfies this for 
$\tilde \sigma \le \sigma \le 0.53$.)
We show that to prove this assumption, it is sufficient to devise an algorithm that computes a rectangular max-min matrix product roughly as efficiently as ordinary 
$(+, \cdot)$ matrix product.

Our algorithms heavily exploit recent constructions of directed shortcuts~\cite{KoganParterDiShortcuts1, KoganParterDiShortcuts2}.

\end{abstract}

\tableofcontents

\section{Introduction}\label{sec:Intro}
\subsection{Our Unconditional Results}\label{sec:UncondResults}
Consider an $n$-vertex directed graph $G = (V,E)$ and a pair of vertices $u, v \in V$.
We say that $v$ \emph{is reachable from} $u$ (denoted $u \leadsto v$) iff there exists a directed path  in $G$ that starts at $u$ and ends in $v$.
Given a subset $S \subseteq V$ of $n^{\sigma}$ designated vertices called \emph{sources}, for some $0 \le \sigma \le 1$, the $S \times V$-\emph{reachability matrix} $R =R(S, V)$ is a rectangular Boolean matrix of dimension $n^{\sigma} \times n$, in which for every pair $(s,v) \in S \times V$, the respective entry $R[s,v]$ is equal to $1$ iff $v$ is \emph{reachable} from $s$ in $G$.
The $S \times V$-\emph{direachability problem} is, given a directed graph $G = (V,E)$, and a subset $S \subseteq V$ of sources, to compute the reachability matrix $R(S,V)$.

To the best of our knowledge, despite the fundamental nature of this problem, it was not explicitly studied so far.
There are two existing obvious solutions for it.
The first one runs a BFS (or DFS) from every source $s \in S$, and as a result requires $O(m \cdot n^{\sigma})$ time.
We call this algorithm \emph{BFS-based}.
For dense graphs this running time may be as large as $O(n^{2 + \sigma})$.
The second solution (which we call \emph{TC-based}) uses fast square matrix multiplication to compute 
the transitive closure $TC(G)$ in $\tilde O(n^{\omega})$ time, where $\omega \le 2.371552\ldots$ is the \emph{matrix multiplication exponent}.
The upper bound $\omega \le 2.371552\ldots$ is due to~\cite{VassilevskaBestRectangular} (see also~\cite{gall2023faster, LeGallBestRectangular, LeGallBestSquare,COPPERSMITH}).\footnote{The expression $O(n^{\omega})$ is the running time of the state-of-the-art algorithm for computing a matrix product of two square matrices of dimensions $n \times n$.}

In this paper we employ recent breakthroughs due to Kogan and Parter~\cite{KoganParterDiShortcuts1, KoganParterDiShortcuts2} concerning directed shortcuts, 
and devise an algorithm that outperforms the two existing solutions (BFS-based and TC-based) for computing $S \times V$-reachability matrix 
whenever $\tilde \sigma \le \sigma \le 0.53$.
(Recall that $\sigma = \log_{n} |S|$. Here $1/3 < \tilde \sigma < 0.3336$ is a universal constant.)
The running time of our algorithm is $\tilde O(n^{1 + \tiny{ \frac{2}{3} }\omega(\sigma)})$, where $\omega(\sigma)$
is \emph{rectangular matrix multiplication exponent}, i.e., $O(n^{\omega(\sigma)})$ is the running time of the state-of-the-art algorithm for
 multiplying a rectangular matrix with dimensions $n^{\sigma} \times n$ by a square matrix of dimensions $n \times n$.
See Table~\ref{tab:tableLeGall} for the values of this exponent (due to~\cite{VassilevskaBestRectangular}),
and Table~\ref{tab:tableComparisonDense} for the comparison of running time of our algorithm 
with that of previous algorithms for $S \times V$-direachability problem in the range 
$\tilde \sigma \le \sigma \le 0.53$ (in which we improve the previous state-of-the-art). 
The largest gap between our exponent and the previous state-of-the-art is for $0.37 \le \sigma \le 0.38$. 
Specifically, our exponent for $\sigma = 0.38$ is $2.33751$, while the state-of-the-art is $2.371552$.

Moreover, we show that our algorithm improves the state-of-the-art solutions for $S \times V$-direachability problem 
in a much wider range of $\sigma$ on sparser graphs. 
Specifically, we show that for every $\tilde \sigma \le \sigma < 1$, there exists a non-empty interval $I_{\sigma} \subset [1,2]$ of values $\mu$,
so that for all input graphs with $m = \Theta(n^{\mu})$ edges, our algorithm outperforms both the BFS-based and TC-based
solutions. The running time of our algorithm on such graphs is $\tilde O(n^{ \tiny{\frac{1 + \mu + 2\cdot \omega(\sigma)}{3}}})$.

We provide some sample values of these intervals in Table~\ref{tab:table1}. 
See also Table~\ref{tab:tableComparisonMu} for the comparison between our new bounds for $S \times V$-direachability on sparser graphs
 and the state-of-the-art ones.
For example, for $\sigma = 0.5$ (when $|S| = \sqrt n$), the interval is $I_{\sigma} = (1.793, 2]$, i.e., our algorithm improves
the state-of-the-art solutions as long as $\mu > 1.793$.
Specifically, when $\mu = 1.9$, our algorithm requires $\tilde O(n^{2.3287\ldots})$ time for computing reachability from $\sqrt n$ sources,
 while the state-of-the art solution is the TC-based one, and it requires $\tilde O(n^{2.371552\ldots})$ time.
 Another example is when $\sigma = 0.6$.
 The interval $I_{\sigma} = I_{0.6}$ is $(1.693, 1.93)$, and for $\mu = 1.75$ our solution requires $\tilde O(n^{2.312\ldots})$ time,
 while the state-of-the-art BFS-based solution requires $O(n^{2.35})$ time.
 \subsection{Our Conditional Results}\label{sec:CondResults}
 We also devise a recursive variant of our algorithm that further improves our bounds \emph{conditioned} 
 on a certain algorithmic assumption. Specifically, we define the \emph{directed reachability} problem
 with parameters $n^{\sigma}$ and $n^{\delta}$ denoted $DR(n^{\sigma}, n^{\delta})$, for a pair of parameters $0 \le \sigma, \delta \le 1$.
 The parameter $\sigma$ is $\log_{n} |S|$ as above.
 The parameter $\delta$ affects the \emph{hop-bound} $D = n^{\delta}$. 
 The problem $DR(n^{\sigma}, n^{\delta})$ accepts as 
 input an $n$-vertex graph $G = (V,E)$ and a set $S$ of $n^{\sigma}$ sources, 
 and asks to compute, for every $s \in S$, the set $\mathcal V_s$
 of vertices reachable from $s$ in at most $n^{\delta}$ hops.
 For a problem $\Pi$, we denote by $\mathcal T(\Pi)$ its time complexity.
 In particular $\mathcal T(DR(n^{\sigma}, n^{\delta}) )$ stands for time complexity of $DR(n^{\sigma}, n^{\delta})$.
 Observe that $DR(n^{\sigma}, n^{\delta})$ can be easily solved by $n^{\delta}$ rectangular Boolean matrix products of an $n^{\sigma} \times n$ matrix
 by an $ n \times n$ square matrix. 
 Let $A = A(G)$ be the adjacency matrix of the input graph $G$.
 In the first iteration, the rectangular matrix $B$ is just the matrix $A$ restricted to $n^{\sigma}$ rows
 that correspond to sources $S$, 
 while the square matrix $A'$ is the adjacency matrix $A$ augmented with a self-loop for every vertex.
 On the next iteration $B$ is replaced by the Boolean matrix product $B \star A'$, etc.
 This process continues for $n^{\delta}$ iterations.
 Thus, $\mathcal T (DR(n^{\sigma}, n^{\delta})) = O(n^{\omega(\sigma) + \delta })$.
 Now consider the following \emph{paths directed reachability problem}, $PDR(n^{\sigma}, n^{\delta})$, 
 in which instead of $n^{\sigma}$ designated source vertices
  we are given $n^{\sigma}$ dipaths $\mathcal P = \{p_1, p_2, \ldots, p_{n^{\sigma}} \}$.
 For every path $p \in \mathcal P$, and every vertex $u \in V$, we want to compute if $u$ is reachable 
 from some vertex $v \in V(p)$ in at most $n^{\delta}$ hops, and if it is the case, 
 we want to output the last vertex $v_p(u)$ on $p$ from which $u$ is reachable within at most $n^{\delta}$ hops.
 
 Interestingly, the BFS-based algorithm for $DR(n^{\sigma}, n^{\delta})$ problem applies (with slight modifications) to the 
 $PDR(n^{\sigma}, n^{\delta})$ problem, providing a solution with running time $O(n^{\mu + \sigma})$ (like for the $DR(n^{\sigma}, n^{\delta})$ problem)~\cite{KoganParterDiShortcuts2}.
 On the other hand, it does not seem to be the case for the algorithm with running time $O(n^{\omega(\sigma) + \delta})$ time for 
 $DR(n^{\sigma}, n^{\delta})$ that we described above.
Our assumption, which we call \emph{paths direachability assumption}, 
is that one nevertheless can achieve running time $\tilde O(n^{\omega(\sigma) + \delta})$ for  $PDR(n^{\sigma}, n^{\delta})$ problem as well.
 
 We show that under this assumption a variant of our algorithm improves the state-of-the-art $S \times V$-reachability algorithms for dense graphs ($\mu = 2$) for all values of $ \tilde \sigma \le \sigma <1$ (and not just in the range $\tilde \sigma \le  \sigma \le 0.53$, which we show unconditionally). Recall that $1/3 < \tilde \sigma < 0.3336$ is a universal constant.
 
 We analyse a sequence of algorithms $\mathcal A_0, \mathcal A_1, \mathcal A_2,\ldots$.
 For each algorithm $\mathcal A_i$, $i = 0,1,2,\ldots$,
 we consider the threshold value $\sigma_i <1$ upto which this algorithm outperforms
 the existing state-of-the-art solutions.
 We show that $0.53  \le \sigma_0 < \sigma_1 < \ldots <\sigma_i < \sigma_{i + 1} < \ldots \le 1$, 
 and that $\lim_{i \to \infty} \sigma_i =1$.
 In other words, for any exponent $\sigma <1$ of the number of sources,
 there exists a (sufficiently large) index $i$ such that the algorithm $\mathcal A_i$ outperforms
 existing state-of-the-art solutions for the $S \times V$-direachability problem with $|S| = n^{\sigma}$.
 %
 
 We also identify some additional assumptions under which the paths direachability assumption holds.
 In particular, consider the \emph{max-min} matrix product.
 Given two matrices $B$ and $A$ of appropriate dimensions (so that matrix product $B \cdot A$ is defined),
 max-min product $C = B \ovee A$ is defined as follows:
 for indices $i,j$ so that $i$ is an index of a row of $B$ and $j$ is an index of a column of $A$, we have
 $C[i,j] = \max_{k} \min \{B[i,k], A[k,j] \}$, where $k$ ranges over all possible indices of columns of $B$.
 
 We show that  $PDR(n^{\sigma}, n^{\delta})$ problem reduces to $n^{\delta}$ invocations of the max-min matrix products 
 for matrices of dimensions $n^{\sigma} \times n$ and $n \times n$.
 Thus, assuming that such a product can be computed in $\tilde O(n^{\omega(\sigma)})$ time implies
 paths direachability assumption.
 We refer to this assumption as \emph{max-min product} assumption.
 Max-min matrix product was studied in~\cite{minMaxProductDP} and~\cite{GrandoniAllPairs}.
 However, current state-of-the-art bounds for it are far from the desired bound of $\tilde O(n^{\omega(\sigma)})$.

 Admittedly, it is not clear how likely are these assumptions.
 We find it to be an intriguing direction for future research.
Also, we believe that the analysis of our algorithms under these assumptions is of independent interest, and that it sheds light on the complexity of multi-source reachability problem.
Furthermore, getting \emph{sufficiently} close to the running time of $\tilde O(n^{\omega(\sigma) + \delta})$ for paths direachability assumption or 
 to $\tilde O(n^{\omega(\sigma)})$ for the max-min product assumption would suffice
 for providing further improvements to the state-of-the-art bounds for this fundamental problem.
 
  \subsection{Our Techniques}\label{par:recurMotivation}
 In this section we provide a high-level overview of our algorithms.
  Given a digraph $G = (V,E)$, and a positive integer parameter $D$,
  we say that a graph $G' = (V,H)$ is a \emph{$D$-shortcut} of $G$ if
  for any ordered pair $u,v \in V$ of vertices, $v$ is reachable from $u$ in $G$
  iff it is reachable from $u$ in $G \cup G' = (V, E \cup H)$ using at most $D$ hops.
  Kogan and Parter~\cite{KoganParterDiShortcuts1, KoganParterDiShortcuts2} showed that
  for any parameter $1 \le D \le \sqrt n$, for any digraph $G$, there exists a $D$-shortcut $G'$ with
  $\tilde O(\tiny{\frac{n^2}{D^3}} + n)$ edges, and moreover, this shortcut can be constructed in $\tilde O(m \cdot n/D^2 + n^{3/2})$ time~\cite{KoganParterDiShortcuts2}.
  
  Our first algorithm starts with invoking the algorithm of~\cite{KoganParterDiShortcuts2} for an appropriate parameter $D$.
 We call this the \emph{diameter-reduction step} of our algorithm.
  We then build two matrices.
  The first one is the adjacency matrix $A'$ of the graph $G \cup G'$, with all the diagonal entries equal to $1$.
  This is a square $ n \times n$ matrix, and the diagonal entries correspond to self-loops.
  The second matrix $B$ is a rectangular $|S| \times n$ matrix.
  It is just the adjacency matrix of of the graph $G \cup G'$ restricted to $|S|$ rows corresponding
  to designated sources of $S$.
  
  Next, our algorithm computes the Boolean matrix product $B \star A'^{D}$.
  Specifically, the algorithm computes the product $B \star A'$, and then multiplies it by $A'$, etc., and does so $D$ times.
  Hence, this computation reduces to $D$ Boolean matrix products of a rectangular matrix of dimensions $|S| \times n$ by a square 
  matrix with dimensions $n \times n$. Each of these $D$ Boolean matrix products can be computed using standard matrix multiplication (see Section~\ref{sec:bmm}).
    We use fast rectangular matrix multiplication algorithm by~\cite{VassilevskaBestRectangular} to compute these products.
  As $|S| = n^{\sigma}$, this computation, henceforth referred to as the \emph{reachability computation step}, requires $O(D \cdot n^{\omega(\sigma)})$ time.
  
  Observe that the running time of the reachability computation step grows with $D$, and the running time of
  the diameter-reduction step decreases with $D$. It is now left to balance these two running times by carefully choosing $D$.
  This completes the overview of our basic direachability algorithm, that improves (unconditionally)
  the state-of-the-art running time for $S \times V$-direachability, $|S| = n^{\sigma}$, for $\tilde \sigma \le \sigma \le 0.53$, $\tilde \sigma < 0.3336$.
  Moreover, as was mentioned above, for any $\sigma$, $\tilde \sigma \le \sigma < 1$, there is a non-empty interval $I_{\sigma} \subseteq [1,2]$ of 
  values $\mu$, such that thus algorithm improves (also, unconditionally) upon state-of-the-art solutions for this problem on graphs with $m = \Theta(n^{\mu})$ edges,
  
  To improve these results further (based on the paths direachability assumption),
  we observe that the algorithm of Kogan and Parter~\cite{KoganParterDiShortcuts2} for building shortcuts has two parts.
  In the first part it computes a subset $P'$ of $\tilde O(n/D^2)$ dipaths and a subset $V'$ of $\tilde O(n/D)$ vertices.
  This step is implemented via a reduction to a min-cost max-flow (MCMF) instance. Using the recent advances in this area~\cite{ChenMaxFlow22},
 it can be done within $\tilde O(m^{1 + o(1)} + n^{3/2})$ time.
  The second part of the algorithm of~\cite{KoganParterDiShortcuts2} involves computing reachabilities between all pairs $(p, v) \in P' \times V'$.
  Here for every pair $(p,v) \in P' \times V'$, one needs to compute the last (if any) vertex on $p$ from which $v$ is reachable.
  Under our paths direachability assumption, this second part can be implemented via a recursive invocation of our $S \times V$-reachability algorithm.
  Our conditional bounds are achieved by analysing this recursive scheme.
  
  \textbf{Related Work}
  Fast rectangular matrix multiplication in conjunction with hopsets was recently employed in a similar way in the context of
  distance computation in \emph{undirected} graphs by Elkin and Neiman~\cite{ElkinN22}.
  Directed hopsets, based on Kogan-Parter constructions of shortcuts~\cite{KoganParterDiShortcuts1, KoganParterDiShortcuts2}, 
  were recently devised in~\cite{BernsteinWein}.
\section{Preliminaries}\label{sec:prelims}
\textbf{Graph Notation.} For an $n$-vertex, $m$-edge digraph $G = (V,E)$, let
 $TC(G)$ denote its transitive closure.
For a vertex pair $u, v \in  V$, where $(u, v) \in TC(G)$, let $d_G(u, v)$
be the length (measured by the number of arcs) of a shortest dipath from $u$ to $v$. For
$(u, v) \notin T C(G)$, we have $d_G(u, v) = \infty$. 
The \emph{diameter} of $G$, denoted $Diam(G)$, is defined as 
$\max_{(u,v) \in TC(G)} d_G(u,v)$. An edge $(u,v) \in TC(G)$ which is not in $E$ is called a \emph{shortcut} edge.

For a dipath $p$, let $V(p)$ denote the vertex set of $p$.
Also, for two distinct vertices $u,v \in V(p)$, we write $u <_p  v$ if $u$ appears before $v$ in $p$, and write $v <_p u$, otherwise.
We use this notation for any vertex sequence $p$, i.e., even if $p$ is not necessarily a dipath.

As mentioned in Section~\ref{sec:Intro}, a naive solution to the $S \times V$-direachability involves performing a BFS or DFS exploration 
from each of the source vertices in $S$.
\begin{definition}\label{def:naiveReach}
For an $n$-vertex $m$-edge digraph $G = (V,E)$, we refer to the naive methods for $S\times V$-direachability (which involve running a BFS or DFS separately
 from each of the sources in $S$) as $S \times V$-\emph{naiveReach} method.
We denote by $\mathcal{T}_{Naive}$ the time complexity of  $S \times V$-naiveReach, and it is given by $\mathcal{T}_{Naive} = O(m \cdot |S|)$.
\end{definition}

Another way to solve $S \times V$-direachability is to compute the transitive closure of the input digraph using square matrix multiplication.
\begin{definition}\label{def:squareReach}
For an $n$-vertex $m$-edge digraph $G = (V,E)$, we refer to the technique of computing the transitive closure of $G$ using 
square matrix multiplication as $S \times V$-\emph{squareReach}. Observe that computing the transitive closure of $G$ 
solves $S \times V$-direachability for up to $|S| = n$ sources.
We denote by $\mathcal{T}_{Square}$ the time complexity of  $S \times V$-squareReach, and it is given by $\mathcal{T}_{Square} = \tilde O(n^{\omega})$,
where $\omega$ is the exponent of $n$ in the number of operations required to compute the product of two $n \times n$ square matrices.
\end{definition}
\begin{definition}\label{def:shortcut}
For a digraph $G$, a set of shortcut edges $H \subseteq TC(G)$ is called a $D$-shortcut if $Diam(G \cup H) \le D$.
\end{definition}
The following theorem summarizes the fast shortcut computation algorithm by Kogan and Parter~\cite{KoganParterDiShortcuts2}.
\begin{theorem}\label{thm:koganParter1}
 (Theorem 1.4,~\cite{KoganParterDiShortcuts2}) There exists a randomized algorithm that for every
$n$-vertex $m$-edge digraph $G$ and $D = O(\sqrt n)$, computes, w.h.p., in time $\tilde O(m \cdot n/D^2 + n^{3/2})$
a $D$-shortcut set $H \subseteq T C(G)$ with $|E(H)| = \tilde O(n^2/D^3 + n)$ edges.
\end{theorem}
\textbf{Matrix Notation.}
For a matrix B, we denote by $B[i,j]$ the entry in row $i$ and column $j$ of $B$. 
The transpose of $B$ is $B^T$ . 
We use $*$ to denote a wildcard, e.g., the notation $B[*,j]$ refers to the vector
which is the $j$-th column of $B$. 

\subsection{Matrix Multiplication and Boolean Matrix Product}\label{sec:bmm}
Recall that for a fixed integer $n$ and $0 < \sigma < 1$, $\omega(\sigma)$ denotes the exponent
of $n$ in the number of algebraic operations required to compute the product of an $n^{\sigma} \times n$ matrix by an $n \times n$ matrix.
Let $B$ be an $n^{\sigma} \times n$ Boolean matrix (i.e., each entry of $B$ is either $0$ or $1$.).
Let $A$ be another Boolean matrix of dimensions $n \times n$. Define the Boolean matrix product $C = B \star A$ by
\begin{equation}\label{def:BMM}
		C[i,j] =  \small{\bigvee_{1 \le k \le n}} B[i,k] \wedge A[k,j],
\end{equation}
where $\vee$ and $\wedge$ stand for binary operations $OR$ and $AND$ respectively.

Note that the Boolean matrix product (BMM henceforth) of $B$ and $A$ can be easily computed by treating the two matrices as standard integer matrices and computing 
the integer matrix product $\tilde C = B \cdot A$ over the integers, and then setting $C[i,j] = 1$ if $\tilde C[i,j] \ge 1$. 
Thus the number of operations required to compute the BMM of $B$ and $A$ as above 
is $O(n^{\omega{(\sigma)}})$.
Vassilevska Williams et al. presented the current state-of-the-art upper bounds on $\omega(\sigma)$ (Table 1 of~\cite{VassilevskaBestRectangular}).
We present these bounds in Table~\ref{tab:tableLeGall} here for completeness.)
\begin{table}[!h]
 \captionsetup{font=scriptsize}
  \begin{center}  
  {\small  
    \begin{tabular}{l|c} 
     $\sigma$ & \textbf{upper bound on $\omega(\sigma) $}\\
      \hline
                0.321334 & 2\\
                0.33 & 2.000100\\
                0.34 &  2.000600\\
                0.35 & 2.001363\\
                0.40 & 2.009541\\
                0.45 & 2.023788\\
        \end{tabular}
        \hspace{1em}
         \begin{tabular}{l|c}
          $\sigma$ & \textbf{upper bound on $\omega(\sigma) $}\\
          \hline
            0.50 & 2.042994\\
            0.527661 & 2.055322\\  
            0.55 & 2.066134\\
             0.60 & 2.092631\\
             0.65 & 2.121734\\
              0.70 & 2.153048\\                               
          \end{tabular}
          \hspace{1 em}        
      \begin{tabular}{l|c}
          $\sigma$ & \textbf{upper bound on $\omega(\sigma) $}\\
          \hline
                 0.75 & 2.186210\\               
                0.80 & 2.220929\\
                0.85 & 2.256984\\
                0.90 & 2.294209\\
                0.95  &2.332440\\
                1.00 & 2.371552\\
         \end{tabular}  
         }                  
  \end{center}
  \caption{Upper bounds on the exponent of $n$ in the number of operations required to 
multiply  an $n^{\sigma} \times n$  matrix by an $n \times n$ 
matrix (reproduced from~\cite{VassilevskaBestRectangular} here for completeness).} 
   \label{tab:tableLeGall}  
\end{table}

\input{BasicTwoStep}
\input{RecursionDense}
\input{RecursionSparse}

\bibliographystyle{alphaurl}
\input{DireachDraft.bbl}

\end{document}

%% file: BasicTwoStep.tex
\section{Basic Two-Step  \texorpdfstring{ $S \times V$}{[S times V]}-Direachability}\label{sec:BasicTwoStep}
In this section, we present the description and analysis of our two-step reachability  scheme.
Let $G = (V,E)$ be a directed unweighted graph.
Let $S \subseteq V$ be a set of source vertices such that $|S| = n^{\sigma}$ for some $0< \sigma < 1$.
We compute $S \times V$-direachability by executing Algorithm~\ref{alg:nonrecursiveDireach}.
We will henceforth refer to it as \emph{Basic $S \times V$-Direachability Algorithm}.
The algorithm accepts as input the graph $G$ and the set $S$ of sources.
In addition, it accepts as input a parameter $D$. The algorithm is described for an arbitrary choice of $D$.
However, in fact, we will set it as a certain function of $n$ and $|S|$, which will be explicated in the sequel.
(See Equation~\eqref{eq:deltaNoRecursion2}.)
\begin{algorithm}[h!]
   {\small
   \caption{\small{DiReach($G,S, D$)}}
  	\label{alg:nonrecursiveDireach}
 		\begin{algorithmic}[1]
			\State \text{Compute a} $D$\text{-shortcut} $H$ \text {for} $G$ \text{(see Theorem~\ref{thm:koganParter1}) }; \label{ln:Step1}
			\State \text{Define a matrix} $A' =A + I$ \text{,where $A$ is the adjacency matrix of} $G \cup H$; \label{ln:Step2}
			\State \text{Let} $B^{(1)} = A_{S*}$; \label{ln:Step3} \Comment{Matrix $A$ restricted to the rows corresponding to vertices in the set $S$.}
		 	\For{$k$ \text{from} $1$ \text{to} $D-1$}\label{ln:Step4}
        				\State \text{Compute} $B^{(k+ 1)} = B^{(k)} \star A'$;
     			 \EndFor\label{ln:Step5}
		 	\State \textbf{return} $B^{(D)}$;		 	
		\end{algorithmic}
		}
\end{algorithm}
The lines~\ref{ln:Step1} and~\ref{ln:Step2}  constitute the diameter reduction step.
We compute a $D$-shortcut $H$ (line~\ref{ln:Step1}) of the input digraph $G$ using~\cite{KoganParterDiShortcuts2} 
(see Theorem~\ref{thm:koganParter1})
and define an $n \times n$ boolean matrix $A' = A + I$, where $A$ is the adjacency matrix of $G \cup H$ (line~\ref{ln:Step2}).
In line~\ref{ln:Step3} we define a new matrix $B^{(1)}$,
which is the adjacency matrix of $G \cup H$ restricted to the rows corresponding to the set $S$ of source vertices.
Lines~\ref{ln:Step4} to~\ref{ln:Step5} constitute the reachability computation step.
We repeatedly perform a rectangular boolean matrix product, and return the final product as our output.

Let $\mathcal T_{DR}$ and $\mathcal T_{RC}$ denote the time complexity of the diameter reduction step and reachability computation step, respectively.
Then the overall time complexity of our algorithm is $\mathcal T_{DR} + \mathcal T_{RC}$.
\subsection{General Graphs}\label{sec:basicDense}
In this section, we consider general $n$-vertex, $m$-edge
digraphs, i.e., $m$ may be as large as $n(n-1)$.
In the next section, we analyse the case of sparser graphs, and show that when $m =\Theta(n^{\mu})$, for some $\mu <2$,
our bounds can be further improved.

Let $D = n^{\delta}$ for some $0 < \delta \le 1/2$. 
In the diameter reduction step we compute a $D$-shortcut of $G$.
 By Theorem~\ref{thm:koganParter1}, this can be done in time $ \mathcal T_{DR} = \tilde O(m \cdot n/D^2 + n^{3/2}) = \tilde O(n^{3-2\delta})$.
In the reachability computation step we perform $D =  n^{\delta}$ iterations, each of which computes a rectangular matrix product of 
an $n^{\sigma} \times n$ matrix by an $n \times n$ matrix.
Each matrix product requires $O(n^{\omega(\sigma)})$ time.
It follows that $\mathcal T_{RC} = O(n^{ \omega(\sigma) + \delta } )$.
The overall time complexity of the algorithm is therefore $\mathcal T_{DR} +  \mathcal T_{RC} = \tilde O(n^{3-2\delta} + n^{ \omega(\sigma) + \delta })$.
Let $g_0(\sigma)$ denote the exponent of $n$ (as a function of $\sigma$) in the overall time complexity of our algorithm.
It follows that
\begin{equation}\label{eq:deltaNoRecursion1}
		g_0(\sigma) = \min_{\delta} \max \{3 -2\delta, \omega(\sigma) + \delta \}.
\end{equation}
Since $3-2\delta$ decreases when $\delta$ grows and $\omega(\sigma) + \delta$ grows,
the minimum is achieved when the following equation holds
\begin{align}\label{eq:deltaNoRecursion2}
	\begin{aligned}
		3-2\delta &= \omega(\sigma) + \delta \text{, i.e.,} \\
		\delta &= 1- \frac{1}{3} \cdot \omega(\sigma).		
	\end{aligned}
\end{align}
The parameter $D$ is therefore set as $D = n^{\delta}$, with $\delta$ given by Equation~\eqref{eq:deltaNoRecursion2}.
It follows from~\eqref{eq:deltaNoRecursion1} and~\eqref{eq:deltaNoRecursion2} that 
\begin{equation}\label{eq:g0Sigma}
g_0(\sigma) = 1 + \frac{2}{3} \cdot \omega(\sigma).
\end{equation}
Recall that the time complexity $\mathcal{T}_{Naive}$ of the naive algorithm is given by $ \mathcal{T}_{Naive} = O(m \cdot |S|) = O( n^{2 +\sigma})$. Hence in the range $\omega -2 = 0.371552 \le \sigma \le 1$, the naive algorithm is no better than computing the transitive closure of the input digraph. The latter requires $\tilde O(n^{\omega})$ time (where $\omega = \omega(1))$.
For our algorithm to do better than the naive algorithm in this range we require
\begin{equation*}
             1 + \frac{2}{3} \cdot \omega(\sigma) < \omega.
\end{equation*}
Incorporating the state-of-the-art upper bounds on $\omega(\sigma)$ from Table~\ref{tab:tableLeGall}, we get that $g_0(\sigma) < \omega$, for
$\sigma \le 0.53$. Therefore, since $\omega(\cdot)$ is a monotone function, 
our algorithm can compute $S \times V$-direachability in $o(n^\omega)$ time for $1 \le |S| \le n^{0.53}$ sources, as opposed to 
naive algorithm that can only do so for $1 \le |S| \le n^{\omega -2}$ .

Moreover, for $0.335 \le \sigma < 0.371552$, $g_0(\sigma) = 1 + \frac{2}{3} \omega(\sigma) < 2 + \sigma$.
Observe that by convexity of $\omega(\sigma)$, it suffices to verify the inequality for the endpoints $\sigma = 0.335$ and $\sigma = 0.371552$.
Therefore in the range $n^{0.335} \le |S| < n^{0.53}$, our algorithm outperforms the state-of-the-art solutions.

Denote by $\tilde \sigma$ the threshold value such that 
\begin{equation}\label{eq:sigmaThresholdTS}
 g_0(\tilde \sigma) = 1 + \frac{2}{3} \omega(\tilde \sigma)  = 2 + \tilde \sigma.
\end{equation}
For all $\tilde \sigma < \sigma <1$, we have 
\begin{equation}\label{eq:sigmaThresholdTS1}
g_0(\sigma) < 2 + \sigma.
\end{equation}
As the Inequality~\eqref{eq:sigmaThresholdTS1} holds for $\sigma = 0.335$, it follows that $1/3 < \tilde \sigma < 0.335$.
(The Inequality~\eqref{eq:sigmaThresholdTS1} does not hold for $\sigma = 1/3$.)
In fact, in Section~\ref{sec:basicSparse} (between Lemma~\ref{lem:thresh} and Lemma~\ref{lem:sparseSigmaZero}) we argue that $\tilde \sigma < 0.3336$.
We conclude with the following theorem:
\begin{theorem}\label{thm: noRecursionDense}
Let $G = (V, E)$ be an $n$-vertex directed unweighted graph with $\theta(n^2)$ edges.
Fix $S \subset V$, $|S| = n^{\sigma}$, for $\tilde \sigma \le \sigma \le 0.53$.
Then, there exists a randomized algorithm that computes with high probability $S \times V$-direachability 
on $G$ in $o(n^{\omega})$ time.
This algorithm outperforms the state-of-the-art solutions for $S \times V$-direachability for 
$n^{\tilde \sigma} \le |S| \le n^{0.53}$.
\end{theorem}
For example, we compute $S \times V$-direachability for $S = n^{\sigma}$, $\sigma = 0.4$ 
in time $O(n^{g_0(0.4)}) = O(n^{2.33969})$, improving upon the state-of-the-art bound of $O(n^{2.4})$ 
 for naive algorithm and the state-of-the-art bound of $\tilde O(n^{2.371552})$ obtained by fast square matrix multiplication.
 In Table~\ref{tab:tableComparisonDense}, we present values of $g_0(\sigma)$ corresponding to some specific values of $\sigma$ in the range $[\tilde \sigma, 0.53]$
 and show how they compare to the exponent of $n$ in $S \times V$-naiveReach (see Definition~\ref{def:naiveReach})  and $S \times V$-squareReach (see Definition~\ref{def:squareReach}).
  \begin{table}[!h]
  \captionsetup{font=scriptsize}
  \begin{center}  
  {\small 
    \begin{tabular}{c|c|c|c} 
     $\sigma$ &$g_0(\sigma)$ &Exponent of $n$ in $\mathcal{T}_{Naive}$ & Exponent of $n$ in $\mathcal{T}_{Square} $\\
      \hline
                0.335 & 2.333565 & \textbf{2.335} & 2.371552 \\
                0.34 & 2.3337 &  \textbf{2.34} & 2.371552\\
                0.35 & 2.334241 &  \textbf{2.35} &2.371552\\
                0.36 &2.33533 &  \textbf{2.36} & 2.371552\\
                0.37& 2.336422 & \textbf{2.37} & 2.371552\\
                0.38 &2.33751& 2.38 & \textbf{2.371552}\\
                0.39 &2.3386 & 2.39& \textbf{2.371552}\\
                0.40 &2.33969 & 2.40 &\textbf{2.371552}\\
                0.41& 2.34159 & 2.41&\textbf{2.371552}\\
                0.42 &2.34349 & 2.42 & \textbf{2.371552}\\
                0.43 & 2.34539 &2.43 &\textbf{2.371552}\\
                0.44 & 2.34729 &2.44 & \textbf{2.371552}\\
                0.45 &2.34919 & 2.45&  \textbf{2.371552}\\
                0.46 & 2.35175 & 2.46 & \textbf{2.371552}\\
                0.47 & 2.35431& 2.47 &  \textbf{2.371552}\\
                0.48& 2.35687 & 2.48 & \textbf{2.371552}\\
                0.49 &2.359435 & 2.49 &  \textbf{2.371552}\\
                0.50 & 2.3621996 & 2.5 & \textbf{2.371552}\\
                0.51& 2.365081& 2.51 & \textbf{2.371552}\\
                0.52 & 2.368166 & 2.52 & \textbf{2.371552}\\
                0.53 & 2.371252 & 2.53 & \textbf{2.371552}                
        \end{tabular}  
        }               
  \end{center}
  \caption{Comparison of our algorithm from Theorem~\ref{thm: noRecursionDense} with the $S \times V$-naiveReach (BFS-based)  and $S \times V$-squareReach (TC-based) methods in terms of exponent of $n$ in their respective time complexities.
 Previous state-of-the-art exponents (for corresponding exponents of $n$ in $S$) are marked by bold font.
 Our exponent (given in the column titled $g_0(\sigma)$) is better than the previous state-of-the-art in the entire range presented in this table (i.e., $0.335 \le \sigma \le 0.53$).
 }    
   \label{tab:tableComparisonDense}  
\end{table}
\subsection{Graphs with  \texorpdfstring{$m = \Theta( n^{\mu})$} {[Theta (n\textsuperscript{mu})]} edges, for  \texorpdfstring{$\mu <2$} {[mu less than 2]} }\label{sec:basicSparse}
In this section we argue that when our input $n$-vertex $m$-edge digraph $G$ has $m = \Theta(n^{\mu})$ edges, for
some $\mu <2$, then our bounds from Theorem~\ref{thm: noRecursionDense} can be further improved.
Let $m = \Theta(n^{\mu})$ be the size of the edge set of $G$, where $0 \le \mu < 2$.
As before,  let $D = n^{\delta}$, for $0 < \delta \le 1/2$.
Applying Theorem~\ref{thm:koganParter1} for $m = \Theta(n^{\mu})$
and $D = n^{\delta}$, it follows that $\mathcal T_{DR} = \tilde O(n^{1 + \mu -2\delta} + n^{3/2})$.
Recall that the running time of the reachability computation step is
$\mathcal T_{RC} = O(n^{\omega(\sigma) + \delta})$.
The second term of $\mathcal T_{DR}$ will always be dominated by $\mathcal T_{RC}$.
The overall time complexity of the algorithm is therefore 
$\mathcal T_{DR} +  \mathcal T_{RC} = \tilde O(n^{1 + \mu -2\delta} + n^{ \omega(\sigma) + \delta })$.
Let $g^{(\mu)}_0(\sigma)$ denote the exponent of $n$ (as a function of $\sigma$) in the overall time time complexity of our algorithm.
It follows that
\begin{equation}\label{eq:deltaNoRecursion1Sparse}
		g^{(\mu)}_0(\sigma) = \min_{\delta} \max \{1 + \mu -2\delta, \omega(\sigma) + \delta \}.
\end{equation}
This expression is minimized when
\begin{align}\label{eq:deltaNoRecursion2Sparse}
	\begin{aligned}
		1 + \mu-2\delta &= \omega(\sigma) + \delta \text{, i.e.,}\\
		\delta &=  \frac{1 + \mu - \omega(\sigma)}{3}.		
	\end{aligned}
\end{align} 
It follows from~\eqref{eq:deltaNoRecursion1Sparse} and~\eqref{eq:deltaNoRecursion2Sparse} that
\begin{equation}\label{eq:g0mu}
    g^{(\mu)}_0(\sigma) = \frac{1 + \mu + 2\omega(\sigma)}{3}.
\end{equation}
Note that for $\mu = 2$, we have $g^{(\mu)}_0 (\sigma) = g_0(\sigma)$.
Also, when $\mu <2$, the exponent $g^{(\mu)}_0(\sigma)$ is smaller than $\omega$ for a wider range of $\sigma$ than
 the exponent $g_0(\sigma)$.
 On the other hand the sparsity of the input digraph also improves the time complexity $\mathcal{T}_{Naive} = O(m \cdot n^{\sigma}) = O(n^{\mu + \sigma})$
 of the naive algorithm. 
 It is possible that for a given combination of $\sigma$ and $\mu$, $\mathcal{T}_{Naive}$ 
 may be better than the time complexity, $\tilde O(n^{\frac{1 + \mu + 2\omega(\sigma)}{3} })$, of our algorithm.
For our algorithm to outperform the naive algorithm, we need
\begin{equation*}
		\frac{1 + \mu + 2\omega(\sigma)}{3} < \mu + \sigma.
\end{equation*}
This condition is equivalent to
\begin{equation}\label{eq:muLower}
		\mu >    \omega(\sigma)  + \frac{1}{2} - \frac{3}{2} \sigma.		
\end{equation}
 By substituting values of $\omega(\sigma)$ from Table~\ref{tab:tableLeGall} for specific values of $\sigma$, we get corresponding lower bounds on $\mu$.
In Table~\ref{tab:table0}, we present these lower bounds for various values of $\sigma$.
\begin{table}[h!]
 \captionsetup{font=scriptsize}
  \begin{center}
   {\small 
    \begin{tabular}{l|c} 
      \textbf{$\sigma$} & \textbf{$\mu $}\\
      \hline
                0 & $>2.5$\\
                0.34 &$> 1.9906$\\
                0.4 & $>1.91$\\
                0.5 &$> 1.793$\\
                0.6& $ > 1.693$ \\
                0.7 &$>  $ 1.603\\
                0.8 & $>1.521$\\
                0.9 & $> 1.444$\\
                1 & $> \omega - 1$
  \end{tabular}
  }
  \end{center}
   \caption{Values of $\mu$ for which $g_0(\sigma) < \mu + \sigma$.
   The condition $\mu > 2.5$ in the first row means that for small values of $\sigma$ state-of-the-art 
   solutions outperform our algorithm for all densities of input graph.}
    \label{tab:table0}
\end{table}

For our algorithm to work in $o(n^{\omega})$ time, we need 
\begin{equation}\label{eq:muUpper}
		\frac{1 + \mu + 2\omega(\sigma)}{3} < \omega \text{, i.e.,~}
		\mu < 3 \omega - 2\omega(\sigma) - 1.
\end{equation}
For $\sigma =1$, Inequality~\eqref{eq:muUpper} implies that $\mu < \omega -1$, matching the lower bound $\mu > \omega -1$ implied by~\eqref{eq:muLower}. Indeed, our algorithm improves existing bounds only for $\sigma <1$.
For every $\sigma < 1$, we get some non-trivial intervals $I_{\sigma}$ of $\mu$ for which $g^{(\mu)}_0(\sigma) < \min \{\mu + \sigma, \omega\}$.
For example, for $\sigma = 0.5$, $\mu > 1.793$ (see Table~\ref{tab:table0}) ensures that $g^{(\mu)}_0(\sigma) < \mu + \sigma$.
For inequality $g^{(\mu)}_0(\sigma)  < \omega$ to hold, substituting $\sigma = 0.5$ in~\eqref{eq:muUpper} implies that
$\mu < 2.028668$. It follows that the inequality always holds for all $\mu$, i.e., the condition on $\mu$ for $\sigma =0.5$ is that $\mu > 1.793$.
We computed these values of intervals for various values of $\sigma$.
Specifically, for a number of values of $\sigma$, using the best known upper bounds on $\omega(\sigma)$, 
we numerically computed the intervals $I_{\sigma} = (\mu_1, \mu_2)$, such that for graphs with $m = \Theta(n^{\mu})$ edges, $\mu_1 < \mu < \mu_2$,
our algorithm improves the existing bounds bounds on the $S \times V$-direachability problem.
In Table~\ref{tab:table1}, we present these intervals for various values of $\sigma$.
\begin{table}[h!]
  \captionsetup{font=scriptsize}
  \begin{center}
  {\small 
    \begin{tabular}{l|c} 
      \textbf{$\sigma$} & \textbf{Interval $I_{\sigma} = (\mu_0, \mu_1)$ }\\
      \hline
      0.335&$ \mu > 1.99785$\\
      0.34 &$\mu > 1.9911$\\
           0.4 & $\mu > 1.91$ \\
      	0.5 & $\mu > 1.793$ \\
      	0.55 & $(1.74, 1.982)$\\
      	0.6 & $(1.693, 1.93)$\\
      	0.7 & $(1.603, 1.809)$\\
      	0.8 & $(1.521, 1.673)$\\
     	 0.9 & $(1.4442, 1.526)$\\
     	 0.99 & $(1.3787, 1.3872)$\\
           1 & $(\omega-1, \omega-1) = \phi$
  \end{tabular}
  }
  \end{center}
    \caption{Intervals of values of $\mu$ for which $g^{(\mu)}_0(\sigma) < \min \{\mu + \sigma, \omega \}. $
                 These intervals are not empty for any $\sigma$, $\tilde \sigma \le \sigma <1$, where $\tilde \sigma <0.335$ is a universal constant.
          }
    \label{tab:table1}
\end{table}
These results suggest that for all $\sigma, \tilde \sigma < \sigma < 1$, there is a non-empty interval of $\mu$ 
such that for digraphs with $m = \Theta(n^\mu)$ edges, our algorithm improves over existing bounds for $S \times V$-direachability
for $|S| = n^{\sigma}$. We will soon prove it formally.

In Table~\ref{tab:tableComparisonMu}, we present a comparison of our algorithm for graphs with $m = \Theta(n^\mu)$ edges
for various combinations of $\mu$ and $\sigma$.
  \begin{table}[!h]
  \captionsetup{font=scriptsize}
  \begin{center}   
  {\small
    \begin{tabular}{c|c|ccc} 
     $\mu$ & $\sigma$ &$g^{(\mu)}_0(\sigma)$ &Exponent of $n$ in $\mathcal{T}_{Naive}$ & Exponent of $n$ in $\mathcal{T}_{Square} $\\
      \hline
                 1.95 & 0.375 & 2.32 & \textbf{2.325} & 2.371552\\
                 &0.4 &  2.323 & \textbf{2.35} & 2.371552\\
                 & 0.45 & 2.3325 & 2.4 & \textbf{2.371552}\\
                 & 0.50 & 2.345 & 2.45 & \textbf{2.371552}\\ \hline
                1.9 & 0.45 & 2.3159 & \textbf{2.35} & 2.371552\\ 
                 & 0.5 & 2.3287 & 2.4 & \textbf{2.371552}\\
                 & 0.55& 2.344 & 2.45& \textbf{2.371552}\\
                 & 0.6 & 2.362 & 2.5 & \textbf{2.371552}\\  \hline
                1.75 & 0.55  &2.294 & \textbf{2.3} & 2.371552\\
                 & 0.6 & 2.312 & \textbf{2.35} &2.371552\\
                & 0.65 & 2.331 &2.4 &\textbf{2.371552}\\
                 & 0.7 & 2.352 & 2.45 &\textbf{2.371552}\\ \hline
                 1.525& 0.8 &2.323 & \textbf{2.325}&2.371552\\
                  & 0.85 & 2.346 & 2.375 & \textbf{2.371552}\\
                 & 0.9 & 2.3711 & 2.425 & \textbf{2.371552} \\ \hline                   
        \end{tabular} 
        }                
  \end{center}
  \caption{Comparison of our algorithm for graphs with $m = \Theta(n^{\mu})$ edges
  with the $S \times V$-naiveReach  and $S \times V$-squareReach methods in terms of exponent of $n$ in their respective time complexities. For each pair $(\mu, \sigma)$ that we list in the table, we have $\mu \in I_{\sigma}$.
 We mark the previous state-of-the-art bounds by bold font.}    
   \label{tab:tableComparisonMu}  
\end{table}

Recall that $\frac{1}{2} + \omega(\sigma) - \frac{3}{2} \sigma$ is the lower bound on the value of $\mu$ implied by Inequality~\eqref{eq:muLower},
i.e., for $\mu > 1/2 + \omega(\sigma) - 3/2 \sigma$, our algorithm is better than the $S \times V$-naiveReach method.
The threshold  $ \tilde \sigma$ is the value such that $\frac{1}{2} + \omega(\sigma) - \frac{3}{2} \sigma =2$, i.e., $\omega( \tilde \sigma) = \frac{3}{2} (1 +  \tilde \sigma)$ holds.
(Recall that we defined it in Section~\ref{sec:basicDense} as the value that satisfies $1 + \frac{2}{3} \omega( \tilde \sigma) = 2+  \tilde \sigma$.
These two conditions are equivalent.)
\begin{lemma}\label{lem:thresh}
For every $\sigma >  \tilde \sigma$, we have 
\begin{align}
\begin{aligned}\label{eq:threshold}
\frac{1}{2} + \omega(\sigma) - \frac{3}{2} \sigma < 2 \text{, i.e., } \omega(\sigma) < \frac{3}{2}(1 + \sigma).
\end{aligned}
\end{align}
\end{lemma}
\begin{proof}
Note that for $\sigma = \tilde \sigma$, we have $\omega(\sigma) = \frac{3}{2} (1 + \sigma)$ and for $\sigma =1$, $\omega =\omega(1) < \frac{3}{2}(1 + \sigma) =3$. The inequality~\eqref{eq:threshold}
now follows by convexity of the function $\omega(\sigma)$.
\end{proof}
To evaluate the value of $\tilde \sigma$ we note that for $\sigma= 0.34$,
we have $\omega(\sigma) = 2.0006$ (see Table~\ref{tab:tableLeGall})
and $\frac{3}{2}(1 + \sigma) = 2.01$.
For $\sigma =1$, we have $\omega(\sigma) = 2.371552$ and $\frac{3}{2}(1 + \sigma) = 3$.
Thus $\omega(\sigma) < \frac{3}{2} ( 1 + \sigma)$ holds for $\sigma = 0.34$ and for $\sigma =1$.
By convexity of the function $\omega(\sigma)$, it follows that this inequality holds for all $0.34 \le \sigma \le 1$.
However, for $\sigma = 1/3$, $\omega(\sigma) > \omega(0.33)>  2.0001 > 2$, whereas $\frac{3}{2} ( 1 + \sigma) =2$.
Therefore, there exists a constant $\tilde \sigma$, $1/3 < \sigma_0 < 0.34$ (satisfying $\omega(\tilde \sigma) = \frac{3}{2} ( 1 + \tilde \sigma)$),
such that the inequality $\omega(\sigma) < \frac{3}{2}  (1 + \sigma)$
(and thus $g_0(\sigma) > \omega(\sigma) + \frac{1-\sigma}{2}$) holds for $\tilde \sigma < \sigma \le 1$.
In fact, one can verify that $1/3 < \tilde \sigma < 0.3336$, as $\omega(0.3336) < 2.00028$, while 
$\frac{3}{2} ( 1+ 0.3336) = 2.0004$.
In the following lemma we formally prove that for $\tilde \sigma < \sigma < 1$, the intervals $I_{\sigma}$ are non-empty.
\begin{lemma}\label{lem:sparseSigmaZero}
There exists a threshold $\tilde \sigma < 0.3336$ such that for every
$\tilde \sigma < \sigma < 1$, we have 
$\frac{1}{2} + \omega(\sigma) - \frac{3}{2} \sigma < 3\omega - 2\omega(\sigma) -1$,
 and there is equality for $\sigma =1$.
 \end{lemma}
\begin{proof}
For $\sigma = \tilde \sigma$, we have $\frac{1}{2} +  \omega(\sigma) - \frac{3}{2}\sigma =2$ and $3\omega - 2\omega(\sigma) -1 >2$
(as $\tilde \sigma < 0.34$, and $\omega(\cdot)$ is a monotonically increasing function).
For $\tilde \sigma < \sigma < 1$, we argue that $\frac{1}{2} +  \omega(\sigma) - \frac{3}{2}\sigma < 3\omega - 2 \omega(\sigma) -1$, i.e., 
$\frac{3}{2}  + 3\omega(\sigma) - \frac{3}{2} \sigma < 3 \omega$.
It is sufficient to prove that
\begin{align*}
	\begin{aligned}
		\frac{1}{2}(1-\sigma) < \omega(1) - \omega(\sigma),
	\end{aligned}
\end{align*}
 or equivalently
 \begin{align}\label{eq:slopeOmega}
	\begin{aligned}
		\frac{\omega(1) - \omega(\sigma) } {1-\sigma} > \frac{1}{2}.
	\end{aligned}
\end{align}
We prove that the Inequality~\eqref{eq:slopeOmega} holds in the following claim.
\begin{claim}\label{clm:slope}
For $\tilde \sigma < \sigma <1$, we have
$\frac{\omega(1) - \omega(\sigma) } {1-\sigma} > \frac{1}{2}$.
\end{claim}
\begin{proof}
Let $\sigma_1 = \alpha =  0.321334$ and $\sigma_2 = 1$.
According to Table~\ref{tab:tableLeGall}, $\omega(\sigma_1) = 2$ and $\omega(\sigma_2) = \omega =  2.371552$.
It follows that 
\begin{align}\label{eq:slope}
	\begin{aligned}
			\frac{\omega(\sigma_2) - \omega(\sigma_1) } {1-\sigma_1} = 0.5475 > 1/2.
	\end{aligned}
\end{align}
Note that $\frac{\omega(1) - \omega(\sigma_1)}{1 - \sigma_1}$ is the slope of the line connecting points $(\sigma_1, \omega(\sigma_1))$ and $(1, \omega(1))$.
By convexity of the function $\omega(\cdot)$, this slope is smaller than the respective slope $\frac{\omega(1) - \omega(\sigma)}{1 - \sigma}$, for any $\sigma_1 < \sigma <1$.
Hence the latter slope is greater than $1/2$ too.
See Figure~\ref{fig:plot} for an illustration.
\begin{figure}[ht!]
 \captionsetup{font=scriptsize}
	\centering
	\fbox{
          \includegraphics[scale=0.15]{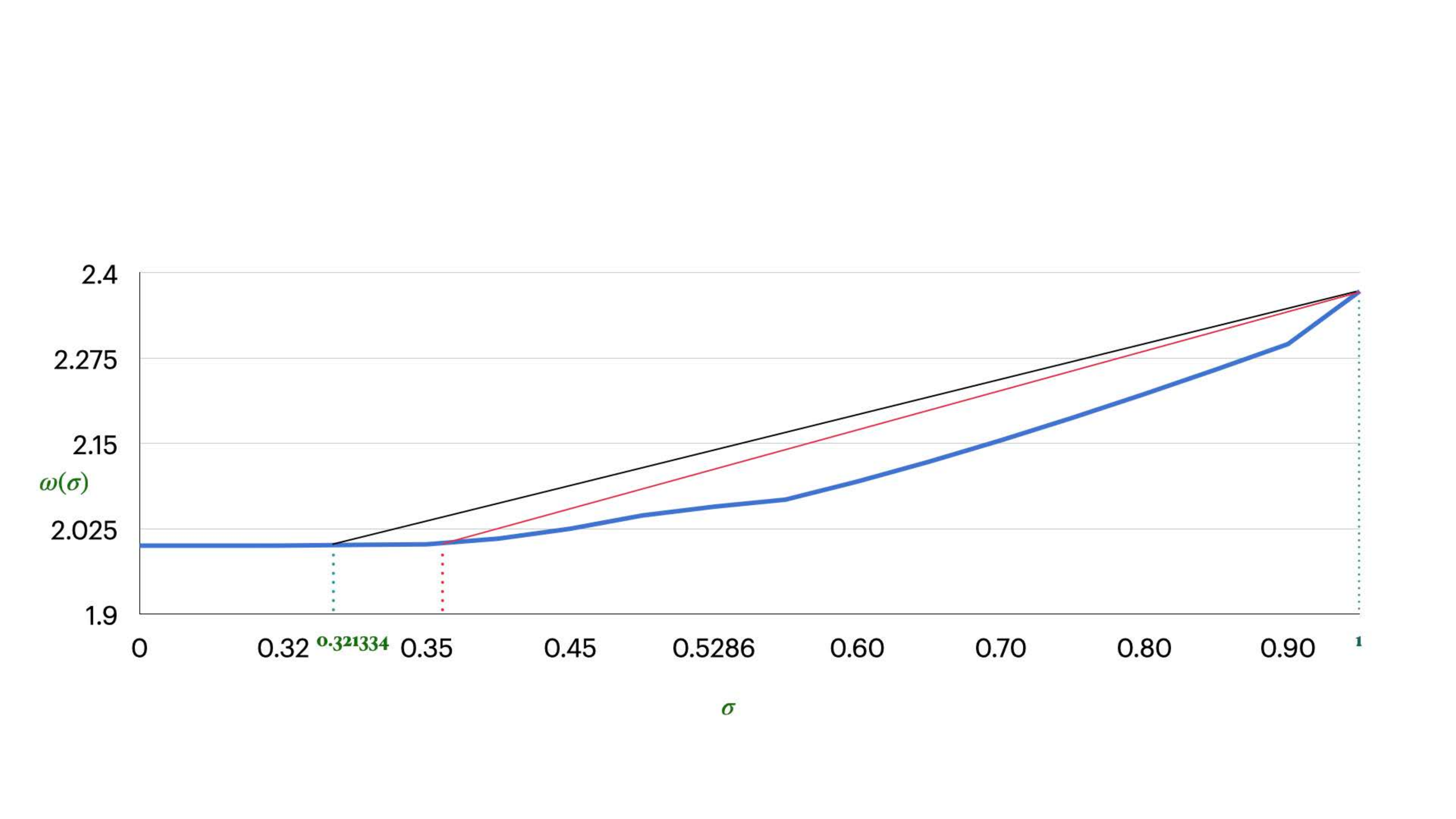}
	}
	\caption{The segment connecting  connecting the points $(0.321334, \omega(0.321334))$ and $(1,\omega(1))$ lies above the plot of $\omega(\sigma)$. Its slope is however smaller than the slope of the segment connecting the points $(\sigma, \omega(\sigma))$ and $(1,\omega(1))$ for any $\tilde \sigma < \sigma < 1$. The latter segment is illustrated by a red line.}
	\label{fig:plot}
\end{figure}

\end{proof}
Hence, the assertion of the lemma holds for $\tilde \sigma < \sigma < 1$.

For $\sigma =1$, we have $\frac{1}{2} + \omega(\sigma) - \frac{3}{2}\sigma = 3\omega - 2\omega(\sigma) -1 = \omega -1$.
Hence equality holds for $\sigma =1$.
\end{proof} 
We proved the following theorem.
\begin{theorem}\label{thm:sparseMain}
There is a threshold value $1/3 < \tilde \sigma < 0.3336$, such that for every $\tilde \sigma < \sigma < 1$, 
there exists a non-empty interval $I_{\sigma}$ that satisfies the following property:\\
For every $\mu \in I_{\sigma}$ and any $n$-vertex $m$-edge digraph with $m = n^{\mu}$ edges,
Algorithm~\ref{alg:nonrecursiveDireach} computes $S \times V$-direachability for $|S| = n^{\sigma}$
faster than the current state-of-the-art algorithms for the problem. 
(The latter are $S \times V$-naiveReach method (see Definition~\ref{def:naiveReach}) and $S \times V$-squareReach method (see Definition~\ref{def:squareReach}).)
See Table~\ref{tab:table1} for sample values of $I_{\sigma}$ and Table~\ref{tab:tableComparisonMu}
for sample exponents of our and state-of-the-art algorithms.
\end{theorem}

%% file: RecursionDense.tex
\section{Recursive  \texorpdfstring{$S \times V$}{[Recursive S times V]}-direachability}\label{sec:recur}
In this section, we describe and analyse our recursive scheme.
The algorithm accepts as input an integer parameter $k = 0, 1,2, \ldots,$ in addition to the input graph $G = (V,E)$,
the set $S \subseteq V$ of sources, and the parameter $D$ that will be set below (as a function of $n$ and $|S|$). 
The algorithm will be referred to as Procedure Recur-DiReach($G,S,D,k$).
If $k =0$, then the algorithm invokes Algorithm~\ref{alg:nonrecursiveDireach}. i.e., Procedure DiReach($G,S,D$), with the parameter 
$D$ set as $n^{\delta}$ with $\delta$ given by Equation~\eqref{eq:deltaNoRecursion2}.
Otherwise for $k \ge 1$, Procedure Recur-DiReach($G,S,D,k$) (like Algorithm~\ref{alg:nonrecursiveDireach})
starts by computing a $D$-shortcut. The difference is, however, that while Algorithm~\ref{alg:nonrecursiveDireach} computes a $D$-shortcut via a blackbox invocation of
Kogan-Parter's algorithm~\cite{KoganParterDiShortcuts2} (see Theorem~\ref{thm:koganParter1}), Procedure Recur-DiReach($G,S,D,k$)
employs itself recursively for computing the shortcut in the following way. It starts with computing the sets $P'$ and $V'$ of paths and vertices, respectively, exactly in the same way as Kogan-Parter's algorithm~\cite{KoganParterDiShortcuts2} does. That is, it computes a collection of paths $P$ of size $n/D$ by a reduction to an instance of min-cost max-flow and then it samples every path in $P$ independently with probability $\Theta(\log n/ D)$ to obtain the set $P'$,
and also samples every vertex $v \in V$ independently with probability $\Theta(\log n/ D)$ to obtain $V'$.
However, to compute shortcuts between paths in $P'$ and vertices in $V'$ it invokes Procedure Recur-DiReach($G,S,D$) recursively with parameters $G, S = P', D', k-1$.
The parameter $D'$ will be set as a certain function of $n$ and $|P'|$ later in the sequel. See Equation~\eqref{eq:parameterD'}.
(Note that the paths become sources for this invocation.)
Finally, once the shortcut set $H$ is computed, Procedure Recur-DiReach($G,S,D,k$) proceeds in the same way as Procedure 
DiReach($G,S,D$) does, i.e., it executes lines~\ref{ln:Step2}-\ref{ln:Step5} of Algorithm~\ref{alg:nonrecursiveDireach}.
The pseudocode of Procedure Recur-DiReach($G,S,D,k$) is given in Algorithm~\ref{alg:recursiveDireach}.\\
\begin{algorithm}[h!]
   {\small
   \caption{\small{Recur-DiReach($G,S, D,k$)}}
  	\label{alg:recursiveDireach}
 		\begin{algorithmic}[1]
			\If{k = 0}
			  \State \text{Invoke DiReach} $(G,S,D)$ \text {(Algorithm~\ref{alg:nonrecursiveDireach}) and return its output};    \label{lnR:Step1}
			  \Else
			   \State \text{Compute the sets $P'$ and $V'$ of paths and vertices, respectively, as in~\cite{KoganParterDiShortcuts2}}
			   
			      \Comment{these sets are needed for computing 
			                       a $D$-shortcut (see~\cite{KoganParterDiShortcuts2})}
			\State \text{Invoke Recur-DiReach($G,P',D',k-1$), where $D'$ will be set in the sequel.}\label{lnR:Step2}
			 \State \text{Add the output of the previous step (a $D'$-shortcut set $H$ of $G$) to $G$}; \label{lnR:Step3}
			\State \text{Execute lines ~\ref{ln:Step2}-\ref{ln:Step5} of Algorithm~\ref{alg:nonrecursiveDireach}
			     and return its output matrix $B^{(D)}$.}
			 \EndIf
		\end{algorithmic}
		}
\end{algorithm}

 For $k \ge 0$, $0< \sigma < 1$, we denote by $g^{(\mu)}_k(\sigma)$ the exponent of $n$ in the number of operations required to
compute \direach by Procedure Recur-DiReach with depth parameter equal to $k$, with $n^{\sigma}$ sources ($|S| = n^{\sigma}$), on graphs with $m = \Theta(n^{\mu})$ edges, for $\mu \le 2$.
We also write $g_k(\sigma) = g^{(2)}_k(\sigma)$ to denote this exponent in the worst-case, i.e., when $\mu=2$.
Recall that for a given $0< \sigma < 1$, we have $g_0(\sigma) = 1 + \frac{2}{3}\cdot \omega(\sigma)$. 
The base step enables us to compute \direach from
up to $O(n^{0.53})$ sources in $o(n^{\omega})$ time (see Theorem~\ref{thm: noRecursionDense}) on general graphs.
Below we overview the recursive step for $k > 0$.\\
\textbf{Recursive Invocation, $k \ge 1$:}
Recall that (see Section~\ref{par:recurMotivation}) the two main computational parts of 
$D$-shortcut computation algorithm of Kogan and Parter~\cite{KoganParterDiShortcuts2},
for a parameter $1 \le D \le \sqrt n$, are
the computation of a set $P$ of vertex-disjoint dipaths in $TC(G)$, and 
the computation of shortcut edges between a set $V' \subset V$, $|V'| = O(n \log n/D)$, of randomly selected vertices
and a set $P' \subset P$, $|P'| = O(n \log n/D^2)$, of randomly selected paths.
The first part can be computed by reducing the computation of the path collection
to an instance of min-cost max-flow problem for which we use the algorithm by Chen et al.~\cite{ChenMaxFlow22}). See~\cite{KoganParterDiShortcuts2} for more details.
The time complexity of Chen et al.'s algorithm~\cite{ChenMaxFlow22} is $\tilde O(m^{1 + o(1)} + n^{3/2})$.
It follows therefore that the computation of the first part requires $\tilde O(m^{1 + o(1)} + n^{3/2})$ time.
For the second part, for every $(v, p) \in V' \times P'$, 
one needs to add a shortcut edge from $v$ to the first vertex reachable from $v$ (if any)
 on $p$. Equivalently, we can reverse the edge orientations and compute reachabilities from $P'$ to 
 all the vertices in $V'$, where for each $p \in P'$ and $v \in V'$ we aim to find the last vertex on $p$ 
 from which $v$ is reachable. For the sake of current analysis, 
 we assume that this can be done by computing
reachability from $|P'|$ sources, where each path is treated as a source.
We further elaborate about this assumption in the next section.
\subsection{Paths Direachability Assumption}\label{sec:Assumptions}
In this section we introduce a number of auxiliary problems, and assumptions about their time complexities.
In the next section we analyse our recursive scheme under these assumptions, and show that
it improves our basic direachability algorithm (described and analysed in Section~\ref{sec:BasicTwoStep}).
 Without loss of generality, we assume that the input graph $G$ to the $D$-shortcut computation algorithm is a DAG. 
(If it not the case, one can compute and contract its strongly connected components in $O(n+ m)$ time, and reduce the problem to DAGs.)
Denote by $\zeta = (v_1, v_2, \ldots, v_n)$ a fixed topological ordering of vertices of $G$.
\begin{definition}\label{def:sourceDreach}
Given an $n$-vertex DAG $G$, and a pair of parameters $0 \le \sigma, \delta \le 1$, and a set $S$ of $|S| = n^{\sigma}$ sources, 
the \emph{direachability problem with diameter $n^{\delta}$}, denoted $DR(n^{\sigma}, n^{\delta})$,
is to compute for every $s \in S$ the set of vertices $\mathcal V_s \subseteq V$ reachable from $s$ via at most $n^{\delta}$ hops.
\end{definition}
We denote by $\mathcal T(DR(n^{\sigma}, n^{\delta}))$ the time complexity of the problem $DR(n^{\sigma}, n^{\delta})$.
More generally, we will use the notation $\mathcal T (\cdot)$ for the time complexities of various problems that we will consider below.
Note that 
\begin{equation}\label{eq:sourceDreach}
	 \mathcal T ( DR(n^{\sigma}, n^{\delta})) = \tilde O(n^{\omega(\sigma) + \delta}).
\end{equation}
We also consider a variant of this problem is which each source is replaced by a (di-)path.
\begin{definition}\label{def:pathDreach}
Given an $n$-vertex DAG $G$, and a pair of parameters $0 \le \sigma, \delta \le 1$,  and a set $\mathcal P'$ of $|\mathcal P'| = n^{\sigma}$ of
(source) dipaths, the \emph{path direachability problem with diameter $n^{\delta}$}, denoted $PDR(n^{\sigma}, n^{\delta})$,
is to compute for every $p \in \mathcal P'$ the set of vertices $\mathcal V_p \subseteq V$, such that each $z \in  \mathcal V_p$ is reachable
from at least one vertex $v \in V(p)$ via at most $n^{\delta}$ hops. Moreover, it is also required to compute, for every 
$z \in \mathcal V_p$, the last vertex $v_p(z) \in V(p)$ on $p$ from which $z$ is reachable within $n^{\delta}$ hops.
\end{definition}
The assumption that we will make in this section is that $PDR(n^{\sigma}, n^{\delta})$ can be computed within
essentially the same time as $DR(n^{\sigma}, n^{\delta})$.
\begin{assumption}(Paths Direachability Assumption)\label{asmp:pathSameAsSource}
\begin{equation}\label{eq:pathSameAsSource}
\mathcal T(PDR(n^{\sigma}, n^{\delta})) =  \tilde O(\mathcal T(DR(n^{\sigma}, n^{\delta}))).
\end{equation}
\end{assumption}
Next, we discuss several other related problems and algorithms for them that
may eventually lead to bounds close to that of Equation~\eqref{eq:pathSameAsSource}.

First, \emph{one-hop paths direachability} problem, denoted $OHPDR(n^{\sigma})$, is the paths direachability problem
with the second parameter $n^{\delta} = 1$, i.e., $OHPDR(n^{\sigma}) = PDR(n^{\sigma}, 1)$.
A slightly more general variant of this problem, which we call \emph{one-hop monotone subsequence direachability} problem,
abbreviated $OHMSDR(n^{\sigma})$, is when each $p \in \mathcal P'$ is a subsequence $\eta = (v_{i_1},v_{i_2}, \ldots, v_{i_t})$, 
for some $t \ge 1$, $i_1< i_2, \ldots < i_t$, of the (fixed) topological ordering $\zeta =(v_1, v_2, \ldots, v_n)$ of vertices of $G$.
We say that $\eta$ is a \emph{$\zeta$-monotone} subsequence.
Note that any dipath in $TC(G)$ is a $\zeta$-monotone subsequence, but the converse in not necessarily true.

A yet more general problem is when each $p \in \mathcal P'$ is a general (not necessarily $\zeta$-monotone) vertex sequence $p = (v_1, v_2, \ldots, v_q)$ for some $q \ge 1$. In this case we call the respective problem \emph{one-hop (general) sequence direachability} problem, abbreviated $OHSDR(n^{\sigma})$.
\begin{lemma}\label{lem:PDRviaOHSDR}
$PDR(n^{\sigma}, n^{\delta})$ can be computed via $n^{\delta}$ applications of $OHSDR(n^{\sigma})$, i.e.,\\
$\mathcal T(PDR(n^{\sigma}, n^{\delta})) \le  O(\mathcal T(OHSDR(n^{\sigma})) \cdot n^{\delta}$).
\end{lemma}
\begin{proof}
We prove by induction on $i$, $1 \le i \le n^{\delta}$, that 
$\mathcal T(PDR(n^{\sigma}, i)) \le C \cdot i \cdot \mathcal T(OHSDR(n^{\sigma}))$, for some universal constant $C >0$.
The induction base is immediate, as $PDR(n^{\sigma}, 1)$ is a special case of $OHSDR(n^{\sigma})$.
Thus, $\mathcal T(PDR(n^{\sigma}, 1)) \le \mathcal T(OHSDR(n^{\sigma}))$.

For the induction step, we consider some $i$, $1 \le i \le n^{\delta}-1$.
By definition, an algorithm for solving $PDR(n^{\sigma}, i)$, for an input set 
$\mathcal P'$ of $n^{\sigma}$ dipaths, returns for every dipath $p \in \mathcal P'$, 
a set of vertices $\mathcal V_p$. For every $z \in \mathcal V_p$, we have a vertex $v_p(z) \in V(p)$
such that $z$ is reachable from $v_p(z)$ via at most $i$ hops, and moreover, $v_p(z)$ is the 
last vertex on $p$ from which $z$ is reachable within at most $i$ hops.

By induction hypothesis we assume that $i$ invocations of an algorithm for 
$OHSDR(n^{\sigma})$ solve $PDR(n^{\sigma}, i)$, i.e., they produce the sets
$\{\mathcal V_p~|~ p\in P' \}$, and for every $z \in \mathcal V_p$, they produce 
the vertex $v_p(z)$ as above.

As we argue below, via one more invocation of an algorithm for $OHSDR(n^{\sigma})$, one can obtain a solution for
$PDR(n^{\sigma}, i +1)$.
Specifically, for every $p \in \mathcal P'$, $p = ( x_1, x_2, \ldots, x_q )$, for some $q \ge 1$,
we sort the set $\mathcal V_p$ (from the induction step $k =i$) such that the vertices $z$ with $v_p(z) = x_1$ appear first, sorted among 
them with respect to the fixed topological ordering $\zeta$, and then appear vertices $z$ with $v_p(z) = x_2$, etc.,
and finally, vertices $z$ with  $v_p(z) = x_q$ appear last.
(Within each such subset, vertices are ordered according to $\zeta$.)
Denote by $\eta(p)$ the resulting vertex sequence.

We invoke an algorithm for $OHSDR(n^{\sigma})$, on the input set of $n^{\sigma}$ sequences $\eta(p)$, for all $p \in \mathcal P'$.
For every vertex $z$ reachable within at most one hop from sequence $\eta(p)$, let $y = v_{\eta(p)}(z)$ be the last vertex in $\eta(p)$ from which $z$ is reachable within at most one hop. 
The algorithm then sets $v_p(z) = v_p(y)$ and returns it.
(This is done for all vertices $z$ reachable within one hop from $\eta(p)$, for every $p \in \mathcal P'$.)
Also for every $p \in P'$, the algorithm returns $\mathcal V_p = \mathcal V_{\eta(p)}$.
(The sets $\mathcal V_{\eta(p)}$ are returned by an algorithm for $OHSDR(n^{\sigma})$ that
receives sequences $\{ \eta(p)~|~p \in P' \}$ as input).

To prove correctness, suppose that a vertex $z \in V$ is reachable within at most $i + 1$ hops from some vertex on a path $p \in \mathcal P'$
and let $v = v_p(z)$ be the last vertex on $p$ that satisfies this condition. Then there is an incoming neighbour $y$ of $z$ which is reachable from $v_p(z)$ within at most $i$ hops. The algorithm for $PDR(n^{\sigma}, i)$ inserts $y$ into $\eta(p)$, and the algorithm for 
$OHSDR(n^{\sigma})$ discovers that $z$ is reachable from $y$, and thus from $\eta(p)$ (within one hop).
Thus it inserts $z$ into $\mathcal V_p$.

Suppose for contradiction that the algorithm returns some vertex $v' \in V(p)$, $v <_p v'$, as $v_p(z)$.
But then, by correctness of the algorithms that we use for $PDR(n^{\sigma}, i)$ and $OHSDR(n^{\sigma})$, it follows that 
there exists an incoming neighbour $y'$ of $z$ which is reachable from $v'$ within $i$ hops and $v' = v_p(y')$.
This is a contradiction to the assumption that $v = v_p(z)$.
Note also that if $z$ is not $(i+1)$-reachable from $p$, then it is not $1$-reachable from $\eta(p)$.
Thus, it will not be included in the output of the algorithm for $PDR(n^{\sigma}, i+1)$.
\end{proof}
This leads to the following assumption.
\begin{assumption}\label{asmp:OHSDRviaMM}
	\begin{equation}
		\mathcal T(OHSDR(n^{\sigma})) = \tilde O(n^{\omega(\sigma)}).
	\end{equation}
\end{assumption}
By Lemma~\ref{lem:PDRviaOHSDR}, Assumption~\ref{asmp:OHSDRviaMM} implies the Paths Direachability Assumption (see Assumption~\ref{asmp:pathSameAsSource}).

Assumption~\ref{asmp:pathSameAsSource} is sufficient for the analysis that we conduct in this section.
Next, we identify a number of possible approaches that may be used to provide upper bounds for $\mathcal T(PDR(n^{\sigma}, n^{\delta}))$ (in the spirit of Assumption~\ref{asmp:pathSameAsSource}).

First, one can reduce $PDR(n^{\sigma}, n^{\delta})$ on graphs with diameter at most $n^{\delta}$
to the exact $S \times V$ distance computation problem in a directed graph with $|S| = n^{\sigma}$ sources, with integer weights
in the range $\{0, 1, \ldots, n-1\}$. Specifically, given an instance $(G = (V, E), \mathcal P')$ of $PDR(n^{\sigma}, n^{\delta})$, 
one adds sources $S = \{s_p~|~ p \in  \mathcal P'\}$, and connects each source $s_p$ to all the vertices of its respective path $p$.
Denote by $G' = (V', E', w)$ the new graph that we construct.
The vertex set $V'$ is given by $V' = V \cup \{s_p~|~ p \in  \mathcal P'\}$.
The edge set $E'$ is given by $E' = E \cup \{ (s_p, v)~|~ v \in V(p), p \in \mathcal P'\}$.

Denote path  $p = (v_0, v_1, \ldots, v_{\ell(p)} )$ ($p \in P'$),
where $\ell(p)$ is the length (number of hops) of the path $p$.
The weight function $w$ is set by assigning weight $\ell(p)$ to the edge $\langle s_p, v_0 \rangle$, weight $\ell(p) -1$ to the edge $\langle s_p, v_1\rangle$, etc.
(Generally, for all $0 \le i \le \ell(p)$, $w(\langle s_p, v_i \rangle )$ is set to $\ell(p) - i$.)
Finally, all original edges of $G$ are assigned weight $0$.

One then invokes an algorithm for $S \times V'$ distance computation problem on the graph $G'$ (with the set $S$ of $n^{\sigma}$ sources). For each pair $(p, z) \in \mathcal P' \times V$, one sets the last vertex $v_p(z)$ on $p$ from which $z$ is reachable in the following way:
If the distance $d_{G'}(s_p, z) = \infty$ then $z$ is not reachable from $p$ at all.
Otherwise, let $i = d_{G'}(s_p, z)$.
Note that $0 \le i \le \ell(p)$.
We then set $v_{\ell(p) - i}$ as $v_p(z)$.
It is easy to verify the correctness of this reduction.
However, to the best of our knowledge, there are currently no known directed exact distance computation algorithms (for multiple sources) that require less than $\tilde O(n^{\omega})$ time.
The state-of-the-art algorithm by Zwick~\cite{APSPDirectedZwick} requires $\tilde O(n^{2.523667})$ time
for directed \emph{all-pairs} shortest paths.

Another possible approach to the the study of complexity of $PDR(n^{\sigma}, n^{\delta})$ is through \emph{maximum witnesses of Boolean matrix multiplication} (abbreviated as MWBMM) problem.
Recall that given two Boolean matrices $A$ and $B$, $ C = B \star A$ denotes their Boolean matrix product (see~\eqref{def:BMM}).
The MWBMM problem asks for every entry $i, j$ of $C$ for which $C[i,j] = 1$, to compute the maximum index $k$ such that
$B[i,k] = A[k,j] = 1$.

One can reduce the one-hop \emph{monotone} sequence direachability problem ($OHMSDR(n^{\sigma})$) to this problem in the following way: define the matrix $B$ as a $\mathcal P' \times V$ incidence matrix, where entry $B[p,z] =1$, if vertex $z$ belongs to 
$V(p)$, and $0$ otherwise. 
The columns of $B$ are ordered according to the fixed topological ordering $\zeta$ of $G$.
We also need another matrix $A$ which is the $V \times V$ adjacency matrix of the input graph $G$. Here both rows and columns are ordered according to ordering $\zeta$.

Consider the Boolean matrix product $C = B \star A$. For every entry $C[p,z]$, $p \in \mathcal P'$, $z \in V$, such that $C[p,z] = 0$,
the vertex $z$ is not reachable from $p$ within one hop. On the other hand $C[p,z] = 1$ iff $z$ is one-hop reachable 
from $p$, and moreover, the largest witness $v$ such that $B[p,v] = A[v, z] =1$ is the last vertex on $p$ 
from which $z$ is one-hop reachable.

This approach is problematic for two reasons.
First, it only solves the one-hop \emph{monotone} sequence direachability problem, as opposed to the one-hop \emph{general} sequence 
direachability problem that we need for proving Assumption~\ref{asmp:pathSameAsSource}.
Second, to the best of our knowledge, there is currently no known algorithm for the $MWBMM$ problem for matrices of dimensions 
$n^{\sigma} \times n$ and $n \times n$ that works in $\tilde O(n^{\omega(\sigma)})$ time.

Finally, the third approach for making progress towards proving (possibly a weaker form of) Assumption~\ref{asmp:pathSameAsSource}
is via max-min matrix product (see~\cite{minMaxProductDP}  and references therein).

Given two matrices $B$ and $A$ (of appropriate dimensions, i.e., the product $B \cdot A$ is defined), the max-min matrix product 
$C = B \ovee A$ is defined in the following way:
The entry $C[i,j]$ (for $i$ and $j$ in the appropriate respective ranges) is given by $C[i,j] = \underset {k} {\max}~\min \{ B[i,k], A[k,j]\}$.

To reduce $OHSDR(n^{\sigma})$ problem to this matrix product (for matrices of dimensions $n^{\sigma} \times n$ and $n \times n$, respectively), 
we define matrices $B$ and $A$ in the following way:
Let $\mathcal Q$ be a set of $n^{\sigma}$ sequences, which serves as an input instance for $OHSDR(n^{\sigma})$.
Rows of $B$ are indexed by sequences $q \in \mathcal Q$ and its columns are indexed by vertices $v \in V$.
Rows and columns of $A$ are indexed by vertices $v \in V$.
The ordering of the columns of $B$ is the same as that of the rows (and columns) of $A$.
For a sequence $q \in \mathcal Q$ and a vertex $v \in V(q)$, we set $B[q,v] = j_q(v)$, where $j_q(v)$ is the index of $v$ in the sequence $q$.
Otherwise (if $v \notin V(q)$), we set $B[q,v] = -\infty$. 
The matrix $A$ is the adjacency matrix of the input graph, in which $A[v,z] = \infty$ iff $\langle v,z \rangle \in E$ and otherwise $A[v,z] = -\infty$.


 Given a pair $q \in \mathcal Q, z \in V$, consider the entry $C[q,z]$ of the product matrix $C = B \ovee A$.
 If $z$ is not one-hop reachable from $q$, then for all $v \in V$, either $B[q,v] = -\infty$ or $A[v,z] = -\infty$, i.e., 
 we have $\min \{B[q,v], A[v,z] \}=  -\infty$. Hence in this case, $C[q,z] = \underset{v} {\max}~\min \{ B[q,v], A[v,z] \} = -\infty$.
 On the other hand, if $z$ is one-hop reachable from $q$, then for every $v \in V(q)$ for which $\langle v, z \rangle \in E$, we have 
 $B[q,v] = j_q(v)$ and $A[v,z] = \infty$, i.e., $\min \{B[q,v] , A[v,z]  \} =  j_q(v)$.
 Also, for every $v \notin V(q)$, we have $B[q,v] = -\infty$,
 implying that $\min \{B[q,v] , A[v,z] \} = -\infty$.
 
 Hence in this case we have 
 $$C[q,z] = \underset {v} {\max}~\min \{B[q,v] , A[v,z] \} = \underset {v \in V(q), \langle v,z \rangle \in E} {\max} j_q(v),$$ as desired.
 We also need to compute maximum $C'[q,z]$ between $C[q,z]$ computed by the above product) and $B[q,z]$, as it is possible that
 $z \in V(q)$, and that its index in $q$ is higher than the largest index of a vertex $y \in V(q)$ from which $z$ is reachable by exactly one hop.
 
 Recall also that each sequence $q$ is associated with a path $p \in P'$, and for every vertex $z$ with finite $B[q,y]$ (not $-\infty$) we store the vertex $v_p(y)$. At this point these values are recomputed along the lines described in the proof of Lemma~\ref{lem:PDRviaOHSDR}. Specifically, if $C'[q,z] = C[q,z] = \underset{v \in V(q), \langle v,z \rangle \in E} \max j_q(v)$, then let $y \in V(q)$ be a vertex such that 
 $C'[q,z] = C[q,z] = j_p(y)$, and we set $v_p(z) = v_p(y)$.
 Otherwise $C'[q,z] = B[q,z]$ and the value of $v_p(z)$ stays unchanged.
 
  Denote by $\mathcal T(MMMP(n^{\sigma}))$ the complexity of max-min matrix product of $n^{\sigma} \times n$ matrix by an $n \times n$ matrix. We have proved that
 \begin{equation*}\label{eq:maxMinTime}
 	   \mathcal T (OHSDR(n^{\sigma})) \le O( \mathcal T(MMMP(n^{\sigma}))).
 \end{equation*} 
 
 Hence by Lemma~\ref{lem:PDRviaOHSDR}, we conclude the following corollary:
 \begin{cor}\label{cor:PDRviaminMax}
      \[
	 	\mathcal T(PDR(n^{\sigma}, n^{\delta})) = O(\mathcal T(MMMP(n^{\sigma})) \cdot n^{\delta}.
	\]
 \end{cor}
 
 Hence the following assumption would suffice for proving the Paths Direachability Assumption (see Assumption~\ref{asmp:pathSameAsSource}).
  \begin{assumption}\label{asmp:MMMMsameAsMM}
  	\[
		\mathcal T(MMMP(n^{\sigma})) = \tilde O(n^{\omega(\sigma)}).
	\]
  
  \end{assumption}
 
However, to the best of our knowledge, there are currently no known bounds on $\mathcal T(MMMP(n^{\sigma}))$ which are
close to $\tilde O(n^{\omega(\sigma)})$. See~\cite{minMaxProductDP, GrandoniAllPairs} for the state-of-the-art bounds for $\mathcal T(MMMP(n^{\sigma}))$.

We note that the state-of-the-art \emph{quantum} algorithm for $MMMP(n)$  (i.e., $\sigma = 1$, that is, for \emph{square} matrices)~\cite{quantumMaxMin} requires $O(n^{2.473})$ time. 
It is plausible that for general $\sigma$, $0 \le \sigma \le 1$, the time complexity of the quantum algorithm for $MMMP(n^{\sigma})$
is not much larger than $\tilde O(n^{\omega(\sigma)})$.
This could lead to some results along the lines that we obtain under Assumption~\ref{asmp:pathSameAsSource}, albeit admittedly not as strong as those that we attain under Assumption~\ref{asmp:pathSameAsSource} (see Section~\ref{sec:recursiveAnalyisis}).

\subsection{Analysis of the Recursive Algorithm under Paths Direachability Assumption}\label{sec:recursiveAnalyisis}
The current state-of-the-art algorithms based on matrix multiplication for the paths direachability problem
do not provide  running time given in Assumption~\ref{asmp:pathSameAsSource}. 
However, we believe that it is of interest to analyse the (hypothetical) situation that such algorithms will eventually be devised. 
It is worth pointing out that using (BFS-based) naive methods one can compute path direachability for $n^{\sigma}$ dipaths 
in the same time as $S \times V$-direachability for $|S| = n^{\sigma}$ sources. i.e., within $O(m n^\sigma)$ time~\cite{KoganParterDiShortcuts2}.
 
Let $D_{k +1} = n^{\delta}$, $\delta  < 1/2$, be the desired diameter of the invocation of Procedure
Recur-DiReach (Algorithm~\ref{alg:recursiveDireach}) with depth parameter $k+1$.
It follows that $|P'| = \tilde O(n^{1 -2\delta})$ and we need to invoke Procedure Recur-DiReach recursively with depth $k$
for \direach with $|S| =   \tilde O(n^{1 -2\delta})$.
Then, the diameter reduction sub-step of the invocation with depth parameter $k+1$ requires time
$\tilde O(n^{g_{k}(1 - 2\delta)})$ on graphs with $\Theta(n^2)$ edges, and more generally, $\tilde O(n^{g_k^{(\mu)}(1 - 2\delta)})$ time
on graphs with $\Theta(n^{\mu})$ edges, for any $1 \le \mu \le 2$.\footnote{Note that the diameter reduction sub-step of the algorithm
of~\cite{KoganParterDiShortcuts2} is applied to the original graph that has $\Theta(n^{\mu})$ edges.
The reachability sub-step is then applied to a possibly denser graph. We remark also 
that the implication that this step requires $\tilde O(n^{g_k(n^{(1-2\delta)}})$ time in general graphs
and $\tilde O(n^{g_k^{(\mu)}(1 - 2\delta)})$ time on graphs with $\Theta(n^{\mu})$
edges also follows, under Assumption~\ref{asmp:pathSameAsSource}, by induction on $k$.
Under Assumption~\ref{asmp:pathSameAsSource}, step~\ref{lnR:Step2} of Algorithm~\ref{alg:recursiveDireach} (Procedure Recur-DiReach)
is indeed a recursive invocation of the same algorithm on the same input graph $G$ with a decremented depth parameter $k$.}
The reachability computation sub-step of this invocation is the same as the reachability computation step of
Algorithm~\ref{alg:nonrecursiveDireach}. By Assumption~\ref{asmp:pathSameAsSource}, 
for a given $ 0 \le \sigma \le 1$, 
this computation requires $D_{k+1} = n^{\delta}$ iterations of fast rectangular matrix  multiplication,
i.e., $\tilde O(n^{\omega(\sigma) + \delta})$ time.

\subsection{Analysis for General Graphs}
We first analyse the recursive scheme for the most general case when $m$ may be as large as $\Theta(n^2)$.
As evident from the above discussion, the following equation defines the function $g_{k + 1} (\cdot)$ that expresses the exponent of $n$ in the running time of an invocation of Procedure Recur-DiReach (Algorithm~\ref{alg:recursiveDireach}) with depth parameter $k+1$
 in terms of
 the exponent $\delta$ of the desired diameter $D_{k+1} = n^{\delta}$ of this invocation and the
 function $g_k(\cdot)$ that expresses the exponent of $n$ in the running time of its recursive invocation (with depth parameter $k$):
 \begin{equation}\label{eq:g_k+1}	
	g_{k + 1} (\sigma) = \min_{\delta} \max \{ g_k( 1- 2\delta),  \omega(\sigma) + \delta \}
\end{equation}
Notice that for $k=1$, 
\begin{align}\label{eq:g1}
	         g_1 (\sigma) &= \min_{\delta} \max \{ g_0( 1- 2\delta),  \omega(\sigma) + \delta \} \nonumber \\
                     &= \min_{\delta} \max \{1 + \frac{2}{3} \cdot \omega(1 -2\delta) ,  \omega(\sigma) + \delta \}.
\end{align}
Note that $ 1 + \frac{2}{3}\omega (1 - 2\delta)$ monotonically does not increase as $\delta$ grows, and $\omega(\sigma) + \delta$ increases.
Also, for $\delta = 0$, we have 
$1 + \frac{2}{3}\omega (1 - 2\delta) = 1 + \frac{2}{3}\omega (1) > \omega(1) > \omega(\sigma) = \omega(\sigma) + \delta$, for any $\sigma <1$.
For $\delta = 1/2$, for any $0 \le \sigma \le 1$, 
we have $1 + \frac{2}{3} \omega(0) = \frac{7}{3} < 2 \frac{1}{2} \le \omega(\sigma) + 1/2 = \omega(\sigma) + \delta$ (since $\omega(0) =2$).
Thus the function $g_1(\sigma)$ is minimized when 
$1 + \frac{2}{3} \cdot \omega(1 -2\delta) =  \omega(\sigma) + \delta$.
We will also soon show that since $\omega(\sigma)$ is a continuous function, so are also the functions $g_0(\sigma), g_1(\sigma), g_2(\sigma), \ldots$. (This can be argued by induction on $k$.)
We will also show that they are all monotonically increasing, for $\sigma > \tilde \sigma$.
Thus, the parameter $\delta$ will be set so that
\begin{equation}\label{eq:parameterD'}
	 g_k(1-2\delta) = \omega(\sigma) + \delta,
 \end{equation}
 and the parameter $D = D_{k+1}$ will be set as $n^{\delta}$.
 
Analysing Equation~\eqref{eq:g1} numerically using the state-of-the-art upper bounds on $\omega(\sigma)$ for various values of $\sigma$ (see Table~\ref{tab:tableLeGall}),
we get that $g_1 (\sigma) < \omega(1) = \omega$, for $\sigma \le 0.66$. 
Recall that $g_0 (\sigma) < \omega(1) = \omega$, for $\sigma \le 0.53$ (see Theorem~\ref{thm: noRecursionDense}).
Thus, with only one step of our recursive scheme, we can increase the number of sources for which we can compute 
\direach in $o(n^{\omega})$ time in general graphs from $|S| = O(n^{0.53})$ to $|S| = O( n^{0.66})$ (under Assumption~\ref{asmp:pathSameAsSource}).
In Table~\ref{tab:tableRecurse1}, we present $g_k(\sigma)$ for various values of $k$ and $\sigma$.
Each row of the table corresponds to a fixed value of $\sigma$ and shows the value of $g_k(\sigma)$ 
for $k =1, 3, 5, 7$ and $9$.
For each $k$, the row corresponding to the first value of $\sigma$ for which $g_k(\sigma) > \omega$ is highlighted.
Empirically, one can see that with each step of recursion, the number of sources from which we can compute
\direach in $o(n^{\omega})$ time in general graphs increases.
For example, for $k = 7$, the number of sources from which we can compute \direach in $o(n^{\omega})$
time is at least $n^{0.96}$.
 \begin{table}[!h]
  \captionsetup{font=scriptsize}
  \begin{center}   
   {\small
    \begin{tabular}{c|c|c|c|c|c|c} 
     $\sigma$ & $g_0(\sigma)$ &$g_1(\sigma)$ &$g_3(\sigma)$ & $g_5(\sigma)$ & $g_7(\sigma)$ & $g_9(\sigma)$\\
      \hline
                0.34 & 2.333733 & 2.333733 & 2.333733 &2.333733  &2.333733& 2.333733\\
                0.36 & 2.33533  &2.333733 & 2.333733  & 2.333733  &2.333733 &2.333733 \\
                0.40 &2.33969 &2.335332 & 2.333733  & 2.333733 & 2.333733 &2.333733 \\
                0.44 & 2.34729& 2.3375132 & 2.333733 & 2.333733 & 2.333733&2.333733\\
                0.48&2.35687&2.339694& 2.335239 & 2.333733&2.333733 &2.333733\\
                0.52 &2.368166&2.3434932 &2.3353324 &2.333733&2.333733& 2.333733\\
                0.54&\fbox{\textbf{2.374075}}&2.3453928&2.3353324&2.333733&2.333733&2.333733\\
                0.56 &$> \omega$& 2.3514334& 2.3353324  &2.333733 &2.333733 &2.333733\\
                0.60 &$> \omega$ &2.3568744 & 2.337513& 2.333733& 2.333733&2.333733\\
                0.64 &$> \omega$&2.365913 &2.3396939 & 2.335239  & 2.333733&2.333733\\
                0.66 &$> \omega$&2.368166 &2.34203&2.3353324 &2.333733&2.333733\\
                0.68 & $> \omega$ &\fbox{ \textbf{2.374075}} &2.342994&2.3353324 &2.333733& 2.333733\\
                0.72& $> \omega$& $> \omega$ &2.3463128 & 2.336308&2.334272 &2.333733\\
                0.76 & $> \omega$ &$> \omega$ &2.353146  &2.341734& 2.335332&2.333733\\
                0.80 &$> \omega$ &$> \omega$ & 2.361996 &2.344272  &2.3395696& 2.3353324\\
                0.84 &$> \omega$ &$> \omega$ & 2.36977 & 2.351433 &2.344272 &2.3397729\\
                0.85 & $> \omega$&$> \omega$ & \fbox{ \textbf{2.374075}} &2.355351&2.346984&2.341734  \\
                0.88&$> \omega$&$> \omega$ & $> \omega$ &2.359773 & 2.3514334 &2.349319 \\
                0.92&$> \omega$&$> \omega$ &$> \omega$ &2.369501& 2.359773& 2.359501 \\
                0.93& $> \omega$&  $> \omega$ & $> \omega$ &\fbox{ \textbf{ 2.374209}}& 2.364429&2.359501\\
                0.96 &$> \omega$& $> \omega$ & $> \omega $& $> \omega$ &2.370262 &2.370262\\
                0.97  &$> \omega$ &$> \omega$ & $> \omega  $& $> \omega$ &\fbox{ \textbf{2.374793}} &2.370262 \\
                0.98 & $> \omega $ & $> \omega$  &$> \omega$   &$> \omega$ & $> \omega$ &  \fbox{ \textbf{ 2.375907}}          
        \end{tabular}  
        }               
  \end{center}
   \caption{$g_k(\sigma)$ for various values of $k$ and $\sigma$. We remark that in each row the values \emph{strictly} decrease when the index $k$ grows, and in each column the values \emph{strictly} increase as $\sigma$ grows.
   In the table some of these values look equal, because (naturally) we are using only a finite number of digits of precision. }
    \label{tab:tableRecurse1}
\end{table}


We next prove that the functions $g_k(\sigma)$ are monotonically non-decreasing and continuous.
\begin{lemma}\label{lem:gMonotonic}
For any $k \ge 0$, the function $g_k(\sigma)$ is continuous and monotonically non-decreasing in the entire range $0 \le \sigma \le 1$,
and $g_k(1) > \omega$.
\end{lemma}
\begin{proof}
The proof is by induction on $k$.
For $k =0$, by Equation~\eqref{eq:g0Sigma} we have $g_k(\sigma) = 1 + \frac{2}{3} \cdot \omega(\sigma)$. 
We also saw that $g_0(1) > \omega$.
Monotonicity and continuity of the function $g_0(\sigma)$ follows from the fact that
$\omega(\sigma)$ is monotonically non-decreasing and continuous.

Assume inductively that $g_0(\sigma), g_1(\sigma), \ldots, g_k(\sigma)$ are all monotonically non-decreasing
and continuous, and that $g_0(1), g_1(1), \ldots, g_k(1) > \omega$.
Recall that (see~\eqref{eq:g_k+1}), 
\begin{equation*}
g_{k + 1} (\sigma) = \min_{\delta} \max \{ g_k( 1- 2\delta),  \omega(\sigma) + \delta \}.
\end{equation*}
In particular, $g_{k+1}(1) = \min_{\delta} \max \{ g_k( 1- 2\delta),  \omega(1) + \delta \}$.
Note that for $\delta =0$, by the inductive hypothesis $g_k(1) > \omega(1)$.
Also by the inductive hypothesis, $g_k(1-2\delta)$ does not increase as $\delta$ grows, while $\omega(1) + \delta$
increases. For any $\delta >0$, $\max \{g_k(1 - 2\delta), \omega(1) + \delta \} \ge \omega(1) + \delta > \omega(1)$.
Hence $g_{k + 1} (1) > \omega$.


Fix two values $0 \le \sigma \le \sigma' \le 1$.
Since $\omega(\cdot)$ is a monotonically non-decreasing function, 
\begin{align}\label{eq:gkMonotone}
	\begin{aligned}
		g_{k+1}(\sigma') &= \min_{\delta} \max \{g_k(1-2\delta), \omega(\sigma') + \delta \} \\
		                           & \ge  \min_{\delta} \max \{g_k(1-2\delta), \omega(\sigma) + \delta \} = g_{k+1}(\sigma),
           \end{aligned}
\end{align}
i.e., $g_{k+1}(\cdot)$ is monotonically non-decreasing function.

By induction hypothesis, we have $g_k(1) > \omega(1) > \omega(\sigma)$, i.e., 
for $\delta = 0$, it holds that $g_k(1-2\delta) > \omega(\sigma) + \delta$.
We now argue that 
\begin{equation}\label{eq:gkMonotoneDeltaHalf}
g_k(0) \le \omega(\sigma) + 1/2,
\end{equation}
i.e, for $\delta = 1/2$, we have $g_k(1-2\delta) \le \omega(\sigma) + \delta$.
\begin{claim}\label{clm:gK0less}
For any $k \ge 0$, $g_k(0) \le 2 \frac{1}{2}$.
\end{claim}
\begin{proof}
For $k =0$, we have
$g_0(0) = 1 + \frac{2}{3} \omega(0) = 2\frac{1}{3} < 2 \frac{1}{2}$.
For $k >0$, assume inductively that $g_k(0) \le 2\frac{1}{2}$.
Then 
\begin{equation*}
g_{k +1}(0) \le \max \{g_k(1 - 2 \cdot \frac{1}{2} ), \omega(0) + 1/2\} = \max\{g_k(0), 2\frac{1}{2} \} = 2\frac{1}{2}.
\end{equation*}
\end{proof}
On the other hand, for any $0 \le \sigma \le 1$, we have $\omega(\sigma) + 1/2 \ge 2 \frac{1}{2}$, implying inequality~\eqref{eq:gkMonotoneDeltaHalf}.

Hence, for $\delta = 0$, we have $g_k(1- 2\delta) = g_k(1) > \omega(\sigma) + \delta = \omega(\sigma)$,
and for $\delta = 1/2$ it holds that 
$g_k(1 - 2\delta) = g_k(0) \le \omega(\sigma) + \delta = \omega(\sigma) + 1/2$.
Thus by continuity and monotonicity of $g_k(\cdot)$ (which is a part of the induction hypothesis), there exists a
value of $\delta$, $0 < \delta \le 1/2$, such that
\begin{equation*}
	g_{k+1}(\sigma) = \min_{\delta} \max \{g_k(1-2\delta), \omega(\sigma) + \delta  \} = g_k(1-2\delta) = \omega(\sigma) + \delta.
\end{equation*}
For fixed $0 \le \sigma \le \sigma' \le 1$, let $\delta$ and $\delta'$, respectively, be the values such that
$g_{k+1}(\sigma) = g_k(1-2\delta) = \omega(\sigma) + \delta$
and $g_{k+1}(\sigma') = g_k(1-2\delta') = \omega(\sigma') + \delta'$.
Then, $g_{k+1}(\sigma') \ge g_{k+1}(\sigma)$ (by Inequality~\eqref{eq:gkMonotone}) 
implies $g_k(1 - 2\delta') \ge g_k(1-2\delta)$.
By the induction hypothesis, $g_k(\cdot)$ is monotonically non-decreasing.
Thus, $1 - 2\delta' \ge 1-2 \delta$, i.e., $\delta' \le \delta$.

Finally, we argue that $g_{k+1}(\cdot)$ is a continuous function.
Let $\sigma$, $0 \le \sigma <1$, be a fixed value.
Since $\omega(\cdot)$ is a continuous function, for any $\epsilon>0$, there exists $\xi >0 $ such that for any $\sigma'$ with
$\sigma < \sigma' < \sigma + \xi$, we have $\omega(\sigma') - \omega(\sigma) < \epsilon$.
Recall that $g_{k+1}(\sigma) = \omega(\sigma) + \delta = g_k(1-2\delta)$, for some fixed $\delta >0$.
Also,
\[
	\omega(\sigma') + \delta \ge g_k(1-2\delta) = \omega(\sigma) + \delta = g_{k+1}(\sigma).
 \]
 Let $\delta'$ be the value that
 minimizes $g_{k+1}(\sigma')$.
 As $\delta' \le \delta$ and $g_{k+1}(\sigma') = \omega(\sigma') + \delta' = g_k(1-2\delta')$,
 we have 
 \begin{align}\label{eq:gDiff}
 	\begin{aligned}
 g_{k+1}(\sigma') -g_{k+1}(\sigma) &= (\omega(\sigma') + \delta') - (\omega(\sigma) + \delta) \le (\omega(\sigma') + \delta) - (\omega(\sigma) + \delta) \\
& = \omega(\sigma') - \omega(\sigma) < \epsilon.
 \end{aligned}
 \end{align}
 This proves that $g_{k+1}(\cdot)$ is continuous to the right of $\sigma$ (for any $\sigma <1$). 
 The proof that it is continuous to the left of $\sigma$, for any $0 < \sigma \le 1$, is symmetric.
\end{proof}
In the next lemma we argue that with each application of the recursion we obtain running time which 
is no worse than the one we had before the current application.
Moreover, we show that the recursion cannot lead to a better exponent than $\omega(\sigma) + \frac{1-\sigma}{2}$.
Note that unlike Lemma~\ref{lem:gMonotonic} which applies in the entire range $0 < \sigma < 1$, the next lemma
applies only for $\tilde \sigma < \sigma <1$.
\begin{lemma}\label{lem:gLowerBound}
For any $\tilde \sigma <\sigma \le 1$ and any for any integer $k \ge 0$, we have
\[
g_0(\sigma) \ge g_1(\sigma) \ge \ldots \ge g_k(\sigma) \ge \omega(\sigma) + \frac{1-\sigma}{2}.
\]
\end{lemma}
\begin{remark}
See Equation~\eqref{eq:sigmaThresholdTS} in Section~\ref{sec:basicDense} for the definition of $\tilde \sigma$.
\end{remark}
\begin{proof}
The proof follows by induction on $k$.
For $k=0$, we require
\begin{equation}\label{eq:g0LowerBound}
g_0(\sigma) \ge \omega(\sigma) + \frac{1 - \sigma}{2}.
\end{equation}
Recall that $g_0(\sigma) = 1 + \frac{2}{3} \cdot \omega(\sigma)$.
The inequality~\eqref{eq:g0LowerBound} (for $\sigma \ge \tilde \sigma$) now follows from Lemma~\ref{lem:thresh}, i.e., from
\begin{align}\label{eq:omegaSigma}
	\omega(\sigma) \le \frac{3}{2} \cdot (1 + \sigma).
\end{align}

Assume inductively that for some $k$, 
\[ 
	g_0(\sigma) \ge  g_1(\sigma) \ge \ldots \ge g_k(\sigma) \ge \omega(\sigma) + \frac{1-\sigma}{2}.
\]
Recall that
 \begin{equation}\label{eq:gkPlusOne}
	 g_{k + 1} (\sigma) = \min_{\delta} \max \{ g_k( 1- 2\delta),  \omega(\sigma) + \delta \}.
\end{equation} 
Fix $\sigma$, $\tilde \sigma < \sigma \le 1$.
For $0 \le  \delta \le \frac{1-\sigma}{2}$, the function $\omega(\sigma) + \delta$ increases as $\delta$ grows but still we have
 $\omega(\sigma) + \delta \le \omega(\sigma) + \frac{1-\sigma}{2}$. Also, by inductive hypothesis, $ g_k(\sigma) \ge  \omega(\sigma) + \frac{1-\sigma}{2}$.
 Moreover, $g_k(1-2\delta) \ge g_k(\sigma)$, because $1-2\delta \ge \sigma$ and, 
  by Lemma~\ref{lem:gMonotonic}, $g_k(\cdot)$ is monotonically non-decreasing.
 Hence $\max \{\omega(\sigma) + \delta, g_k(1-2\delta)\} = g_k(1-2\delta)$ for any $\delta \le \frac{1-\sigma}{2}$.
 The minimum in this range is achieved by setting $\delta = \frac{1 - \sigma}{2}$, i.e., 
 we have 
 $\min_{\delta \le \frac{1- \sigma}{2}} \max\{g_k(1 - 2\delta), \omega(\sigma) + \delta\} = g_k(\sigma) \ge \omega(\sigma) + \frac{1- \sigma}{2}$.
 At this point we still have $g_k(1 - 2\delta) = g_k(\sigma) \ge \omega(\sigma) + \frac{1 - \sigma}{2}$.
 As we increase $\delta$ further, $\omega(\sigma) + \delta$ increases and $g_k(1-2\delta)$ decreases or stays the same.
 Recall that by Lemma~\ref{lem:gMonotonic}, $g_k(\sigma)$ is a monotonically non-decreasing function in the entire range
 $0 < \sigma < 1$. Suppose first that
 \begin{equation}\label{gKLower}
 	g_k(\sigma) > \omega(\sigma) + \frac{1-\sigma}{2}.
 \end{equation}
 Our analysis splits into two sub-cases.
 In the first sub-case, there exists a value $ \frac{1-\sigma}{2}< \delta \le 1/2$ so that $g_k(1-2\delta) = \omega(\sigma) + \delta$.
 Then for $\epsilon = \delta - \frac{1-\sigma}{2} > 0$, we have $g_k(1-2\delta) = g_k(\sigma - 2\epsilon)= \omega(\sigma) + \frac{1-\sigma}{2} + \epsilon$.
 By Lemma~\ref{lem:gMonotonic} $g_k(1-2\delta)$ is a monotonically non-increasing function of $\delta$, while $\omega(\sigma) + \delta$ is a 
 monotonically increasing function of $\delta$.
 Thus, 
 \begin{align*}
 	\begin{aligned}
		g_{k+1}(\sigma) &= \min_{\delta} \max \{g_k(1 - 2\delta), \omega(\sigma) + \delta \} \\
					  &= g_k(\sigma - 2\epsilon) = \omega(\sigma) + \frac{1-\sigma}{2} + \epsilon\\
					  &> \omega(\sigma) + \frac{1-\sigma}{2}	
	\end{aligned} 
 \end{align*}
 In this sub-case, by Lemma~\ref{lem:gMonotonic} (monotonicity of $g_k(\cdot)$) we also have
 $g_{k + 1}(\sigma) = g_k(\sigma - 2\epsilon) \le g_k(\sigma)$.
 On the other hand, if for all $\delta \le 1/2$, we have $g_k(1-2\delta) > \omega(\sigma) + \delta$, then, in particular, 
 we have $g_k(0) > \omega(\sigma) + 1/2 \ge 2\frac{1}{2}$.
 However, by induction hypothesis, $g_k(0) \le g_0(0) = 1 + \frac{2}{3} \omega(0) = \frac{7}{3}$, contradiction.
 
 The second sub-case is $g_k(\sigma) = \omega(\sigma) + \frac{1- \sigma}{2}$. Then by monotonicity of $g_k(1-2\delta)$
 and $\omega(\sigma) + \delta$ (as functions of $\delta$) we have
 \begin{align*}
 	\begin{aligned}
		\min_{\delta \le \frac{1-\sigma}{2}} \max\{g_k(1- 2\delta), \omega(\sigma) + \delta \} = \omega(\sigma) + \frac{1- \sigma}{2} = g_k(\sigma).
	\end{aligned}
 \end{align*}
 For $\delta > \frac{1-\sigma}{2}$, $\max \{ g_k(1- 2\delta), \omega(\sigma) + \delta\} > \omega(\sigma) + \frac{1-\sigma}{2}$,
 and thus 
 \begin{equation*}
	 g_{k+1}(\sigma) = \min_{\delta} \max\{g_k(1- 2\delta), \omega(\sigma) + \delta \} = \omega(\sigma) + \frac{1-\sigma}{2} = g_k(\sigma).
 \end{equation*}

 \end{proof}
 Recall that the \emph{dual matrix multiplication exponent} $\alpha$ is the maximum (in fact, supremum) value such that $\omega(\alpha) = 2$.

 Next, we argue that in the range $\tilde \sigma \le \sigma \le 1$, the functions $g_k(\sigma)$ ($k \ge 0$) are strictly monotonically increasing. 
 (We have already shown in Lemma~\ref{lem:gMonotonic} that these functions are monotonically non-decreasing in the entire range $0\le \sigma \le1$.)
 \begin{lemma}\label{lem:strict:monoGK}
 For any integer $k\ge 0$, the function $g_k(\sigma)$ is strictly monotonically increasing
 in the range $\alpha < \sigma \le 1$.
 Moreover, in the range $\tilde \sigma < \sigma <1$, we also have
 \[g_0(\sigma) > g_1(\sigma) > \ldots g_k(\sigma) > \omega(\sigma) + \frac{1-\sigma}{2}. \]
 \end{lemma}
 \begin{proof}
 We start with proving the second assertion of the lemma.
 The proof is by induction on $k$.
 The induction base holds as $\omega(\sigma)$ is a monotonically increasing function for $\alpha \le \sigma \le 1$,
 and $g_0(\sigma) = 1 + \frac{2}{3} \omega(\sigma)$.
 Moreover, we have seen that for $\sigma > \tilde \sigma$, we have $g_0(\sigma) > \omega(\sigma) + \frac{1-\sigma}{2}$
 (see Inequalities~\eqref{eq:g0LowerBound} and~\eqref{eq:omegaSigma} and Lemma~\ref{lem:thresh}).
 
 Assume inductively that for some $k \ge 0$, the function $g_k(\sigma)$ is strictly monotonically increasing in the range $\tilde \sigma < \sigma \le 1$,
 and $g_k(\sigma) > \omega(\sigma) + \frac{1-\sigma}{2}$.
 Fix a value $\tilde \sigma  < \sigma \le 1$.
 We have 
 \[ g_{k+1} (\sigma) = \min_{\delta} \max \{g_k(1- 2\delta), \omega(\sigma) + \delta \}. \]
 As by Lemma~\ref{lem:gMonotonic}, we have $g_k(1) > \omega(1) > \omega(\sigma)$ and (by Claim~\ref{clm:gK0less}) $g_k(0) \le 2 \frac{1}{2} \le \omega(\sigma) + 1/2$,
 by continuity of functions $g_k(1-2\delta)$ and $\omega(\sigma) + \delta$ (in terms of $\delta$) in the entire range $0 \le \delta \le 1/2$
 (see Lemma~\ref{lem:gMonotonic}), there exists a value $\delta^*$ such that $\omega(\sigma) + \delta^* = g_k(1-2\delta^*)$.
 By monotonicity of the function $g_k(1-2\delta)$ in the entire range of $\delta$ (by Lemma~\ref{lem:gMonotonic} it is monotonically non-decreasing),
 we also conclude that 
 \[g_{k+1}(\sigma) = \min_{\delta} \max\{g_k(1-2\delta) , \omega(\sigma) + \delta \} = g_k(1-2 \delta^*) = \omega(\sigma) + \delta^*.\]
 For $\delta \le \tiny{\frac{1-\sigma}{2}}$, we have $1 -2\delta \ge \sigma > \tilde \sigma$.
 By induction hypothesis, $g_k(1-2\delta)$ is a strictly monotonically decreasing function of $\delta$ in this range.
 Using induction hypothesis again, we recall that 
 \[g_k(\sigma) > \omega(\sigma) + \frac{1-\sigma}{2}.\]
 Hence for $\delta = \frac{1-\sigma}{2}$ we have 
 \[g_{k+1}(\sigma) \le \max\{g_k(\sigma), \omega(\sigma) + \frac{1-\sigma}{2} \} = g_k(\sigma).\]
 By slightly increasing $\delta$ beyond $\tiny{\frac{1-\sigma}{2}}$ but below $\tiny{\frac{1-\tilde\sigma}{2}}$ (so that $\sigma > 1- 2\delta > \tilde \sigma$ and
 so that $g_k(\sigma) > \omega(\sigma) + \delta$),
 we increase $\omega(\sigma) + \delta$ and decrease $g_k(1-2\delta)$.
 Moreover, we can also have 
 \[g_k(1-2\delta) > \omega(\sigma) + \delta > \omega(\sigma) + \tiny{\frac{1-\sigma}{2}}.\]
 (Recall that by Lemma~\ref{lem:gMonotonic}, $g_k(\cdot)$ is continuous.)
 Let $\delta'$ denote such a value of $\delta$.
 Since $\sigma > 1-2\delta' > \tilde \sigma$ and, by induction hypothesis $g_k(\cdot)$ is monotonically increasing for $\sigma > \tilde \sigma$,
 we have 
 \[g_k(\sigma) > \max \{g_k(1- 2\delta'), \omega(\sigma) + \delta' \} = g_k(1- 2\delta') > \omega(\sigma) + \delta' > \omega(\sigma) +  \tiny{\frac{1-\sigma}{2}}.\]
 It follows that $\delta^* > \delta'$.
 (As for all $\hat \delta \le \delta'$, we have $g_k(1- 2 \hat \delta) > \omega(\sigma) +  \hat \delta$,
 and both functions are monotone and continuous in terms of $\delta$ (by Lemma~\ref{lem:gMonotonic}),
 and $g_k(1- 2\delta^*) = \omega(\sigma) + \delta^*$.)
 Hence 
\[g_{k+1}(\sigma) = \omega(\sigma) + \delta^* > \omega(\sigma) + \delta' > \omega(\sigma) + \frac{1-\sigma}{2}, \]
 and also, $g_{k+1}(\sigma)  = g_k(1-2\delta^*) \le g_k(1-2\delta') < g_k(\sigma)$.
 

 Next, we prove that functions $g_0(\sigma), g_1(\sigma), g_2(\sigma), \ldots$ are all strictly monotonically increasing for
 $\alpha < \sigma$. We have already seen that it is the case for $g_0(\sigma) = 1 + \tiny{\frac{2}{3}} \omega(\sigma)$.
 We assume inductively that it is the case for $g_k(\sigma)$ for some $k = 0,1,2,\ldots,$ and prove it for 
 \[g_{k+1}(\sigma) = \min_{0 \le \delta \le 1/2} \max \{ g_k(1-2\delta), \omega(\sigma) + \delta\}.\]
 Observe that for $\delta = 0$, $g_k(1-2\delta) = g_k(1) > \omega(1) = \omega$, and for $\delta \ge \tiny{\frac{1-\alpha}{2}}$, 
 by monotonicity of $g_k(\cdot)$ we have
 \begin{equation}\label{eq:gk-alpha} 
      g_k(1-2\delta) \le g_k(\alpha) \le g_0(\alpha) = 1 +  \frac{2}{3} \omega(\alpha) = \frac{7}{3}.
 \end{equation}
 On the other hand, for $\delta = 0$, (and $\sigma < 1$), 
 \[\omega(\sigma) + \delta = \omega(\sigma) < \omega < g_k(1 - 2\delta) = g_k(1),\]
 and for $\delta \ge \tiny{\frac{1-\alpha}{2}}$, we have by~\eqref{eq:gk-alpha}
 \[
 	\omega(\sigma) + \delta > \omega(\sigma) + \tiny{\frac{1-\alpha}{2}} > \tiny{\frac{7}{3}} \ge g_k(1-2\delta).
 \]
 By continuity and strict monotonicity of $g_k(1-2\delta)$ for $1-2\delta > \alpha$ (by the induction hypothesis),
 the value $\delta(\sigma)$ for which 
$g_{k+1}(\sigma) = \omega(\sigma) + \delta(\sigma) = g_k(1-2\delta(\sigma))$ satisfies therefore $0 < \delta(\sigma) < \tiny{\frac{1-\alpha}{2}}$.

Consider now two values $\alpha \le \sigma < \sigma' \le 1$.
Then
\begin{equation}\label{eq:strictalpha1}
 \omega(\sigma') + \delta(\sigma) > \omega(\sigma) + \delta(\sigma) = g_k(1- 2 \delta(\sigma)).
\end{equation}
 Consider the value $\delta(\sigma')$ for which 
 \begin{equation}\label{eq:strictalpha2}
 g_{k+1}(\sigma') = \omega(\sigma') + \delta(\sigma') = g_k(1- 2\delta(\sigma')).
 \end{equation}
 If $\delta(\sigma') \ge \delta(\sigma)$ then 
 \[ g_{k+1}(\sigma') = \omega(\sigma') + \delta(\sigma') \ge \omega(\sigma') + \delta(\sigma) > \omega(\sigma) + \delta(\sigma) = g_{k+1}(\sigma), \]
 as desired. Henceforth we assume that
 \[ \delta(\sigma) > \delta(\sigma').\]
 Note also that
 \[ 
 	1 - 2\delta(\sigma), 1 -2 \delta(\sigma') > \alpha,
  \]
  i.e., by the induction hypothesis, the function $g_k(\cdot)$ is strictly monotone between $1-2\delta(\sigma')$ and $1- 2\delta(\sigma)$,
 and 
 \[1-2\delta(\sigma') > 1- 2\delta(\sigma).\]

Hence
  \[
  	g_{k+1}(\sigma') = g_k(1 - 2\delta(\sigma')) > g_k(1- 2\delta(\sigma)) = g_{k+1}(\sigma).
  \]
   \end{proof}
 For any integer $k \ge 0$, let $0 < \sigma_k < 1$ be the value such that $g_k(\sigma_k) = \omega(1) = \omega$.
Observe that by Lemma~\ref{lem:gLowerBound}, we have $\tilde\sigma < \sigma_0 \le \sigma_1 \le \sigma_2 \le \sigma_3 \ldots \le 1$.
Hence the sequence $(\sigma_k)_{k =1}^{\infty}$ converges to a limit $\sigma^* \le 1$, i.e., $\lim_{k \to \infty} \sigma_k = \sigma^*$.
Next, we analyze the value of $\sigmalimit$ and show that it is equal to $1$.
This means that under Assumption~\ref{asmp:pathSameAsSource},  sufficiently many recursive invocations of our algorithm lead 
to an algorithm that improves the state-of-the-art directed reachability for general graphs for all $\tilde \sigma < \sigma <1$.
\begin{lemma}\label{lem:sigmaStar}
\[ \sigmalimit = \lim_{k \to \infty} \sigma_k = 1\]
\end{lemma}
\begin{proof}
Applying Lemma~\ref{lem:gLowerBound} to $\sigma^*$, we have
\[g_0(\sigmalimit) \ge g_1(\sigmalimit) \ge \ldots  \ge g_k(\sigmalimit) \ge \ldots \ge \omega(\sigmalimit) + \frac{1- \sigmalimit}{2} .\]
It follows that the sequence $(g_k(\sigma^*))_{k=0}^{\infty}$ converges to a limit greater than or equal to $\omega(\sigma^*) + \frac{1- \sigma^*}{2}$.
Assume for contradiction that $\sigmalimit < 1$.
The following two cases arise:
\begin{case}\label{cas:case1}
\[\lim_{k \to \infty} g_k(\sigmalimit) = \omega(\sigmalimit) + \frac{1 - \sigmalimit}{2}.\]
\end{case}
In this case, for any arbitrarily small $\epsilon > 0$, there exists $k_{\epsilon}$ such that for every $k > k_{\epsilon}$, we have
\[
	\omega(\sigmalimit) + \frac{1 - \sigmalimit}{2} \le g_k(\sigmalimit) < \omega(\sigmalimit) + \frac{1 - \sigmalimit}{2} + \epsilon.
\]
By Claim~\ref{clm:slope}, for any $\tilde \sigma < \sigma < 1$, we also have $\omega(\sigma) + \frac{1 - \sigma}{2} < \omega(1)$. 
We can pick a sufficiently small $\epsilon > 0$ so that

\[
	\omega(\sigmalimit) + \frac{1- \sigmalimit}{2} + \epsilon < \omega(1).
\]
Hence, for $k > k_{\epsilon}$, $g_k(\sigmalimit) < \omega(1)$.

Since $\tilde \sigma < \sigma_0  \le \sigma_k \le \sigmalimit$ 
and $g_k(\cdot)$ is a monotonically non-decreasing function for $\sigma > \tilde \sigma$ (see Lemma~\ref{lem:gLowerBound}), it follows that
\[
	g_k(\sigma_k) \le g_k(\sigmalimit) < \omega(1).
\]
But by definition, we have $g_k(\sigma_k) = \omega(1)$. Hence we get a contradiction in this case.

\begin{case}\label{cas:case2}
There exists some value $\epsilon^* > 0$ such that
 \[
	\lim_{k \to \infty} g_k(\sigmalimit) = \omega(\sigmalimit) + \frac{1 - \sigmalimit}{2} + \epsilon^*.
 \]

\end{case}
If $\epsilon^*$ is such that $ \omega(\sigmalimit) + \frac{1 - \sigmalimit}{2} + \epsilon^* < \omega(1)$, we get a contradiction 
similar to the one obtained in Case~\ref{cas:case1}.
So we assume from this point on that 
\[
\lim_{k \to \infty} g_k(\sigmalimit) = \omega(\sigmalimit) + \frac{1 - \sigmalimit}{2} + \epsilon^* \ge \omega(1).
\]
It is easy to see that $\epsilon^* \le \frac{\sigmalimit}{2}$.
Indeed, otherwise $\lim_{k \to \infty} g_k(\sigmalimit) \ge \omega(\sigmalimit) + 1/2$.
As $\sigmalimit \ge \sigma_9 > 0.98$ (see Table~\ref{tab:tableRecurse1}), it follows that $\lim_{k \to \infty} g_k(\sigmalimit) \ge \omega(0.98) + 1/2 > 2.85$.
On the other hand, for any $k \ge 1$, $g_k(\sigmalimit) \le g_0(\sigmalimit) = 1 + \frac{2}{3} \omega(\sigmalimit) \le 2.59$,
and this is a contradiction.

Hence, for all $k = 0,1,2,\ldots$ and $\delta = \frac{1- \sigmalimit}{2} + \epsilon^*$, by Equation~\eqref{eq:gkPlusOne}, we have
\[
	g_{k+1} (\sigmalimit) \le \max \{ g_k(\sigmalimit - 2\epsilon^*), \omega(\sigmalimit) + \frac{1 - \sigmalimit}{2} + \epsilon^* \}.
\]
Next, we prove that $g_k(\sigmalimit - 2\epsilon^*)$ maximises the right-hand-side.
\begin{claim}\label{clm:gkSigma}
\[
	g_k(\sigmalimit -2\epsilon^*) \ge \omega(\sigmalimit) + \frac{1-\sigmalimit}{2} + \epsilon^*.
\]
\end{claim}
\begin{proof}
 Suppose for contradiction $g_k(\sigmalimit - 2\epsilon^*) < \omega(\sigmalimit) + \frac{1- \sigmalimit}{2} + \epsilon^*$.
Then, since $g_k(\cdot)$ is a continuous function (see Lemma~\ref{lem:gMonotonic}), there exists $0 \le \epsilon' < \epsilon^*$ such that
\[
	g_{k+1}(\sigmalimit) = \omega(\sigmalimit) + \frac{1- \sigma^*}{2} + \epsilon' = g_k(\sigmalimit - 2\epsilon').
\]
But then 
\[
	g_{k+1}(\sigmalimit) < \omega(\sigmalimit) + \frac{1 - \sigmalimit}{2} + \epsilon^* = \lim_{t \to \infty} g_t(\sigmalimit).
\]
But by Lemma~\ref{lem:gLowerBound}, the sequence $( g_t(\sigmalimit))_{t=1}^{\infty}$ is monotonically non-increasing, 
and so $g_{k+1}(\sigmalimit)$ cannot be smaller than the limit of this sequence.
(Recall that we assumed that $\sigmalimit <1$.)
\end{proof}
Therefore, we have $g_k(\sigmalimit - 2\epsilon^*) \ge \omega(\sigmalimit) + \frac{1- \sigmalimit}{2} + \epsilon^* \ge \omega(1) = g_k(\sigma_k)$.
Hence $\sigmalimit - 2\epsilon^* \ge \sigma_k$.
But $\sigma_k$ (weakly) increases as $k$ grows, and it is equal to $\sigmalimit$ in the limit.
On the other hand, there exists a fixed constant $\epsilon^* > 0$ such that (for all $k$) $\sigma_k \le \sigmalimit -2\epsilon^*$.
This is a contradiction.

Under the assumption that $\sigmalimit <1$, we have derived a contradiction in all the cases for which
$\lim_{t \to \infty} g_t(\sigmalimit) = \omega(\sigmalimit) + \frac{1 - \sigmalimit}{2} + \epsilon^*$, for some $\epsilon^* \ge 0$, 
regardless of  whether $\omega(\sigmalimit) + \frac{1 - \sigmalimit}{2} + \epsilon^*$ is less than $\omega(1)$ or at least 
$\omega(1)$.

Therefore we conclude that $\sigmalimit =1$.
In other words, for any $\epsilon > 0$, there exists $k_{\epsilon}$ such that for any $k \ge k_{\epsilon}$,
we have $\sigma_k > 1- \epsilon$, i.e., $g_k(1 - \epsilon) < g_k(\sigma_k) = \omega(1)$.
\end{proof}
It also follows that 
\[ \lim_{k \to \infty} g_k(1) = \lim_{k \to \infty} g_k(\sigma_k) = \lim_{k \to \infty} \omega(1) = \omega  =  \omega(1) + \frac{1-1}{2} ,\]
and also for any $\sigma < \sigma^* = 1$, it holds that
\begin{equation}\label{eq:limitSigma1}
 	\lim_{k \to \infty} g_k(\sigma) \le \lim_{k \to \infty} g_k(1) = \omega(1).
  \end{equation}

%% file: RecursionSparse.tex
\subsection{Analysis for graphs with  \texorpdfstring{$m = \Theta(n^\mu)$}{[Theta(n\textsuperscript{mu})]} edges, for  \texorpdfstring{$\mu <2$} {[mu less than 2]}}
In this section, we discuss the convergence of the sequence 
$g_0^{(\mu)} (\sigma), g_1^{(\mu)} (\sigma), \dots, g_k^{(\mu)} (\sigma)$, for a general value $\mu < 2$.
Recall that for graphs with $m = \Theta(n^{\mu})$ edges, $g_k^{(\mu)}(\sigma)$ is the exponent of $n$
in the number of operations required (under Assumption~\ref{asmp:pathSameAsSource}) for computing $S \times V$-direachability ($|S| = n^{\sigma}$) by our recursive scheme of depth $k$.
By Inequality~\eqref{eq:g0mu} we have
\[g_0^{(\mu)} (\sigma) = \frac{1+\mu}{3} + \frac{2}{3} \omega(\sigma). \]

The following equation expresses the function $g_{k + 1}^{(\mu)} (\cdot)$ in terms of
 the exponent $\delta$ of the desired diameter $D_{k+1}$ and the
 function $g_k^{(\mu)}(\cdot)$.
\[g_{k+1}^{(\mu)} (\sigma) = \min_{\delta} \max \{g_k^{(\mu)} (1-2\delta), \omega(\sigma) + \delta\}. \]

The formula for $g_{k+1}^{(\mu)}(\cdot)$ for $k > 0$ follows, because the running time required to build
a $D$-shortcut is dominated by the time needed to solve the path-direachability problem for $\tilde O(n^{1-2\delta})$ paths
in a graph with $\Theta(n^{\mu})$ edges. 
The latter, by Assumption~\ref{asmp:pathSameAsSource}, can be done in $\tilde O(n^{ g_k^{(\mu)} (1-2\delta) } )$ time. 
(Observe that the recursive step (Line~\ref{lnR:Step2} of Algorithm~\ref{alg:recursiveDireach}) is invoked on the original input graph $G$.
Hence if $G$ has $O(n^{\mu})$ egdes, all recursive invocations are performed on a graph with $O(n^{\mu})$ edges as well.)

The following lemma is an easy lower bound on $g_k^{(\mu)} (\sigma)$.
\begin{lemma}\label{lem: gmuLowerBound}
For any $k = 0,1,2,\ldots,$ and $\sigma$ and $\mu$ that satisfy $1 + \mu > \omega(\sigma)$, we have
\[ g_k^{(\mu)}(\sigma) > \omega(\sigma).\]
\end{lemma}
\begin{remark}
If $\omega(\sigma) \ge 1 + \mu$,
then a naive direachability algorithm that computes all-pairs direachabilities in $O(m \cdot n) = O(n^{1 + \mu})$ time
is at least as fast as any algorithm based on rectangular matrix multiplication.
(The rectangular matrix multiplication requires $\Omega(n^{\omega(\sigma)})$ time.)
\end{remark}
\begin{proof}
The proof follows by induction on $k$.
First, observe that 
\[ g_0^{(\mu)} (\sigma) = \frac{1 + \mu }{3} + \frac{2}{3} \omega(\sigma) > \omega(\sigma),  \]
for $1 + \mu > \omega(\sigma)$.

Assume that for some $k \ge 0$, we have $ g_k^{(\mu)}(\sigma) > \omega(\sigma)$.
By setting $\delta = \frac{1 - \sigma}{2}$, we get
\begin{align}\label{eq:maxsparse}
g_{k+1}^{(\mu)}(\sigma) \le \max\{ \omega(\sigma) + \frac{1- \sigma}{2}, g_k^{(\mu)}(\sigma) \}.
\end{align}
For $\delta < \frac{1- \sigma}{2}$, we have $g_k^{(\mu)}(1 - 2\delta) \ge  g_k^{(\mu)}(\sigma) > \omega(\sigma)$.
The last inequality is by inductive hypothesis.
For $\delta \ge \frac{1-\sigma}{2}$, we have $\omega(\sigma) + \delta > \omega(\sigma)$.
In either case, $\max \{\omega(\sigma) + \delta, g_k^{(\mu)}(1 - 2\delta) \} > \omega(\sigma)$, 
and thus
\[
   g_{k+1}^{(\mu)} (\sigma) = \min_{\delta} \max \{ \omega(\sigma) + \delta, g_{k}^{(\mu)}(1 - 2\delta) \} > \omega(\sigma).    
\]
\end{proof}
We next prove a stronger inequality.
\begin{lemma}\label{lem:sparseStronger}
Assuming $\mu \ge \omega(\sigma) - \frac{3}{2} \sigma + \frac{1}{2}$, and $\sigma \ge \tilde \sigma$, 
for all $k = 0,1,2,\dots$, we have 
\begin{equation}\label{eq:strongerIneqSparse}
	g_k^{(\mu)}(\sigma) \ge \omega(\sigma) + \frac{1-\sigma}{2}.
\end{equation}
We also have 
\begin{equation}\label{eq:strongerIneqSparse2}
	g_0^{(\mu)}(\sigma) \ge g_1^{(\mu)}(\sigma) \ge \ldots \ge g_k^{(\mu)}(\sigma)\ge \ldots
\end{equation}
Strict inequalities in~\eqref{eq:strongerIneqSparse} and~\eqref{eq:strongerIneqSparse2}  hold when $\mu > \omega(\sigma) - \frac{3}{2} \sigma + \frac{1}{2}$.
Moreover, each $ g_{k}^{(\mu)} (\sigma)$ is a
 continuous and monotonically non-decreasing function in the entire range $[0,1]$.
\end{lemma}
\begin{remark}
Recall that the expression $\omega(\sigma) - \frac{3}{2} \sigma + \frac{1}{2}$ is the left boundary of the feasibility interval 
$I_{\sigma}$. This is the interval in which
$g_0^{(\mu)}(\sigma)$ improves over the existing bound provided by naive direachability algorithm (see Inequality~\eqref{eq:muLower}).
\end{remark}
\begin{proof}
The proof is by induction on $k$. The induction base is $k =0$.\\
For $\delta = \frac{1 - \sigma}{2}$, we have $\omega(\sigma) + \delta = \omega(\sigma) + \frac{1-\sigma}{2}$, 
and 
$$\frac{1 + \mu}{3} + \frac{2}{3} \omega(1 - 2\delta) = \frac{1 + \mu}{3} +\frac{2}{3} \omega(\sigma) = g_0^{(\mu)} (\sigma).$$
Hence we need to argue that
$\omega(\sigma) + \frac{1 - \sigma}{2} \le \frac{1 + \mu}{3} + \frac{2}{3} \omega(\sigma)$, i.e., 
\begin{align}\label{eq:omegaSigmaUpper}
\omega (\sigma) \le \mu + \frac{3}{2} \sigma - \frac{1}{2}.
\end{align}
This is exactly the assumption of the lemma.
Note also that strict inequality in~\eqref{eq:omegaSigmaUpper} implies that
$g_0^{(\mu)} (\sigma) > \omega(\sigma) + \frac{1- \sigma}{2}$.
(For $\mu = 2$, the right-hand-side of~\eqref{eq:omegaSigmaUpper} is $\frac{3}{2}(1 + \sigma)$.
The inequality $\omega(\sigma) \le \frac{3}{2} (1 + \sigma)$ holds for $\sigma \ge \tilde \sigma$. 
See Inequality~\eqref{eq:threshold} in Section~\ref{sec:basicSparse}.)
Note also that since $\omega(\cdot)$ is continuous and monotonically non-decreasing in the range $[0,1]$, so is $g_0^{(\mu)}(\sigma) = \frac{1 + \mu}{3} + \frac{2}{3}\omega(\sigma)$ as well.\\
Similarly, since $\omega(\cdot)$ is monotonically increasing in the range $[\tilde \sigma, 1]$, this is the case with $g_0^{(\mu)}(\sigma)$ too.
For the induction step we again use $\delta = \frac{1-\sigma}{2}$. It follows that
\begin{equation}\label{eq:indHyp}
	\begin{aligned}
    	 	g_{k+1}^{(\mu)}(\sigma) &= \min_{\delta} \max \{g_k^{(\mu)}(1- 2\delta), \omega (\sigma) + \delta \} \\
     		&\le \max \{g_k^{(\mu)} (\sigma), \omega(\sigma) + \frac{1- \sigma}{2}\} = g_k^{(\mu)} (\sigma) \\
	\end{aligned}
\end{equation}
(The last equality is by induction hypothesis.)
Hence, 
\begin{equation}\label{eq:sparseSeqMonotone}
  g_0^{(\mu)}(\sigma) \ge g_1^{(\mu)}(\sigma) \ge \dots, g_k^{(\mu)}(\sigma) \ge \ldots,
\end{equation}
and strict inequality holds if strict inequality holds in~\eqref{eq:omegaSigmaUpper}.

Also, for all $\delta < \frac{1- \sigma}{2}$, by induction hypothesis, $g_k^{(\mu)} (\cdot)$ is monotonically increasing.
Thus, we have
\[
	g_k^{(\mu)}(1 - 2\delta) > g_k^{(\mu)}(\sigma) \ge \omega(\sigma) + \frac{1 - \sigma }{2} > \omega(\sigma) + \delta.
\]
If $g_k^{(\mu)} (\sigma) = \omega(\sigma) + \frac{1-\sigma}{2}$, then since $g_k^{(\mu)} (1 - 2\delta)$ is a monotonically
non-increasing function of $\delta$ and $\omega(\sigma) + \delta$ is a monotonically increasing function of $\delta$,
it follows that 
\[g_{k + 1}^{(\mu)} (\sigma) = \min_{\delta} \max \{ \omega(\sigma) + \delta, g_{k }^{(\mu)}(1 - 2\delta)  \} = \omega(\sigma) + \frac{1-\sigma}{2}. \]
Otherwise, if $g_k^{(\mu)}(\sigma) > \omega(\sigma) + \frac{1-\sigma}{2}$, then there exists $\delta^* > \frac{1- \sigma}{2}$ such that
$g_{k + 1}^{(\mu)} (\sigma) = \omega(\sigma) + \delta^* = g_k(1-2\delta^*) > \omega(\sigma) + \frac{1- \sigma}{2}$, proving that
$g_{k + 1}^{(\mu)} (\sigma) \ge \omega(\sigma) + \frac{1 - \sigma}{2}$ in both cases.
Moreover, this argument implies that if inductively, $g_k^{(\mu)}(\sigma) > \omega(\sigma) + \frac{1- \sigma}{2}$, then 
$g_{k + 1}^{(\mu)} (\sigma) > \omega(\sigma) + \frac{1-\sigma}{2}$ as well.

Finally, the inductive argument that shows that $g_{k + 1}^{(\mu)} (\cdot)$ is continuous and monotonically non-decreasing function of 
$\sigma$ in the range $[0,1]$ 
 is identical to the proof that it is the case for $g_{k+1}(\sigma) = g_{k+1}^{(\mu =2)}(\sigma)$, given in the proof
of Lemma~\ref{lem:gMonotonic}.
\end{proof}


We have seen that for any $\tilde \sigma < \sigma < 1$, there exists a non-empty interval $I^{(0)}_\sigma$ of values of $\mu$ such that for all $\mu \in I^{(0)}_\sigma$ we have
\[
     g_0^{(\mu)} (\sigma) < \min \{ \mu + \sigma, \omega(1) \}.
 \]
Since, for any $\sigma$ in this range we have 
\[
g_0^{(\mu)} (\sigma ) \ge g_1^{(\mu)} (\sigma) \ge \ldots > g_k^{(\mu)} (\sigma) > \ldots,
\]
it follows that these intervals for $g_k^{(\mu)}(\sigma)$ become only wider as $k$ grows, i.e., 

\[
     I^{(0)}_\sigma \subseteq  I^{(1)}_\sigma \subseteq \ldots  I^{(k)}_\sigma \subseteq \ldots.
\]
In fact, the left boundaries of all these intervals all agree. as they correspond to the same inequality 
\begin{align}\label{eq: muLowerSparse}
\mu > \frac{1}{2} + \omega(\sigma) - \frac{3}{2} \sigma.
\end{align}
(To see it, observe that assuming inductively that $g_k^{(\mu)}(\sigma) < \mu + \sigma$ for $\mu > \omega(\sigma) + \frac{1}{2} - \frac{3}{2}\sigma$, implies that $g_{k+1}^{(\mu)}(\sigma) < g_k^{(\mu)}(\sigma) < \mu + \sigma$ under the same condition as well. See Lemma~\ref{lem:sparseStronger} and Inequality~\eqref{eq:sparseSeqMonotone}.)
Recall that both inequalities $g_0^{(\mu)} (\sigma) < \mu + \sigma$ and $\omega(\sigma) + \frac{1 - \sigma}{2} < \mu + \sigma$ both lead to Inequality~\eqref{eq: muLowerSparse}.
However, as $k$ grows, the right boundaries of these intervals tend to $1$.

Recall that the right boundaries of these intervals are determined by the inequalities $g_k^{(\mu)} (\sigma) < \omega(1)$.
These inequalities hold for all $\tilde \sigma < \sigma < \sigma_k$, and by Lemma~\ref{lem:sigmaStar}, $\lim_{t \to \infty} \sigma_t = 1$.

Note also that for $\mu = \omega(\sigma) - \frac{3}{2} \sigma + \frac{1}{2}$, all these functions agree.
\begin{lemma}\label{recursionSparseLast}
For $\mu = \omega(\sigma) - \tiny{\frac{3}{2}} \sigma + \tiny{\frac{1}{2}}$, we have
	\begin{equation}
    		 g_0^{(\mu)} (\sigma) = g_1^{(\mu)} (\sigma) = \ldots = g_k^{(\mu)} (\sigma) = \ldots = \omega(\sigma) + \frac{1 - \sigma}{2}.
     	\end{equation}
\end{lemma}
\begin{proof}

\begin{align*}
	\begin{aligned}
   g_0^{(\mu)} (\sigma) & = \frac{1 + \mu}{3} + \frac{2}{3} \omega(\sigma) \\
                                     & = \frac{ 1 + (\omega(\sigma) - 3/2 \cdot \sigma + 1/2)}{3} + \frac{2}{3} \omega(\sigma) \\
                                    & = \omega(\sigma) + \frac{1 - \sigma}{2}.
   \end{aligned}
\end{align*}

Also, assuming inductively that $g_k^{(\mu)}(\sigma) = \omega(\sigma) + \frac{1 - \sigma}{2}$, we derive that
\begin{align*}
	\begin{aligned}
   g_{k + 1}^{(\mu)} (\sigma) &= \min_{\delta} \max \{\omega(\sigma) + \delta, g_k^{(\mu)} (1 - 2\delta) \} \\
                                             & \le \max \{\omega(\sigma) + \frac{1 - \sigma}{2}, g_{k}^{(\mu)} (\sigma)  \} \\
                                             & =  \omega(\sigma) + \frac{1 - \sigma}{2}.   
   \end{aligned}
\end{align*}
Also, by Lemma~\ref{lem:sparseStronger} for $\sigma \ge \tilde \sigma$, we have $ g_{k + 1}^{(\mu)} (\sigma) \ge \omega(\sigma) + \frac{1 - \sigma}{2}$, and thus
$g_{k + 1}^{(\mu)} (\sigma) = \omega(\sigma) + \frac{1-\sigma}{2}$.
\end{proof}

%% file: DireachDraft.bbl
\newcommand{\etalchar}[1]{$^{#1}$}